\documentclass[11pt,final]{scrartcl}

\usepackage[left=1in,top=1in,right=1in,bottom=1in]{geometry}

\usepackage[utf8]{inputenc}

\usepackage{amsmath,amssymb,amsthm}

\newtheorem{theorem}{Theorem}[section]

\newtheorem{lemma}[theorem]{Lemma}
\newtheorem{proposition}[theorem]{Proposition}

\newtheorem{claim}{Claim}
\newtheorem{corollary}[theorem]{Corollary}

\theoremstyle{definition}
\newtheorem{definition}[theorem]{Definition}

\usepackage{mathtools}
\usepackage{tikz}
\usepackage{graphicx}

\usetikzlibrary{shapes}
\usetikzlibrary{backgrounds}
\usetikzlibrary{decorations.pathmorphing}
\usetikzlibrary{decorations.pathreplacing}
\usetikzlibrary{calc}
\usetikzlibrary{fit}

\usepackage[obeyDraft]{todonotes}
\usepackage[final,margin]{fixme}

\usepackage{bbding} 

\usepackage{ifdraft}

\newcommand{\RR}[1]{\ifdraft{\marginpar{\textcolor{green!80!black}{Roman:}#1}}{}}
\newcommand{\RRl}[1]{\ifdraft{\reversemarginpar\marginpar{\textcolor{green!80!black}{Roman: }#1}\par\normalmarginpar}{}}

\usepackage{xspace}
\usepackage[shortlabels]{enumitem}
\usepackage{mathdots} 

\usepackage{cleveref}

\newdimen\arrowsize
\newlength{\arrowlength}
\newlength{\arrowangle}
\newlength{\arrowthickness}

\pgfarrowsdeclare{slim}{slim}
{
\arrowsize=0.7pt
\advance\arrowsize by .5\pgflinewidth
\pgfarrowsleftextend{-4\arrowsize-.5\pgflinewidth}
\pgfarrowsrightextend{.5\pgflinewidth}
}
{
\advance\arrowsize by .5\pgflinewidth
\pgfsetdash{}{0pt} 
\pgfsetroundjoin 
\pgfsetroundcap 
\pgfpathmoveto{\pgfpointpolar{\arrowangle}{\arrowlength}}
\pgfpathlineto{\pgfpointorigin}
\pgfusepathqstroke
\pgfpathmoveto{\pgfpointpolar{\arrowangle}{\arrowlength}}
\pgfpathcurveto{\pgfpointpolar{\arrowangle*1.04}{\arrowthickness+0.5*(\arrowlength-\arrowthickness)}}%
{\pgfpointpolar{\arrowangle*1.1}{\arrowthickness+0.3*(\arrowlength-\arrowthickness)}}
{\pgfpointpolar{180}{\arrowthickness}}
\pgfpathlineto{\pgfpointorigin}
\pgfusepathqfill
\pgfpathmoveto{\pgfpointpolar{-\arrowangle}{\arrowlength}}
\pgfpathlineto{\pgfpointorigin}
\pgfusepathqstroke
\pgfpathmoveto{\pgfpointpolar{-\arrowangle}{\arrowlength}}
\pgfpathcurveto{\pgfpointpolar{-\arrowangle*1.04}{\arrowthickness+0.5*(\arrowlength-\arrowthickness)}}%
{\pgfpointpolar{-\arrowangle*1.1}{\arrowthickness+0.3*(\arrowlength-\arrowthickness)}}
{\pgfpointpolar{180}{\arrowthickness}}
\pgfpathlineto{\pgfpointorigin}
\pgfusepathqfill
}

\tikzstyle{vertex}=[circle,inner sep=2.5,minimum size =2mm,semithick,fill=black!20, draw=black]
\tikzstyle{holevertexB}=[circle,inner sep=-0.7,minimum size =2mm,semithick,fill=white, draw=black]
\tikzstyle{holevertexA}=[circle,inner sep=2.5,minimum size =2mm,semithick,fill=white, draw=black]
\tikzstyle{smallcircle}=[circle,inner sep=1.5,fill=white, draw=black]
\tikzstyle{point}=[circle,inner sep=1,fill=black, draw=black]
\tikzstyle{path}=[-slim,thin,rounded corners] 
\tikzstyle{path1}=[-slim,thin,decorate,%
decoration={snake,amplitude=6pt,segment length=#1,%
pre length=#1/20,post length=#1/20}]
\tikzstyle{path2}=[-stealth,thin,decorate,%
decoration={snake,amplitude=6pt,segment length=#1,%
pre length=#1/10,post length=#1/10}]
\tikzstyle{brace}=[thin,decorate,decoration=brace]
\tikzstyle{ie}=[thin,dashed,gray]

\colorlet{fillA}{gray!50}
\colorlet{fillB}{gray!15}

\input{abbrs}

\DeclareMathOperator{\Reach}{Reach}

\DeclareMathOperator{\flap}{cmpt}

\DeclareMathOperator{\supp}{supp}

\renewcommand{\mid}{\ensuremath{:}}

\renewcommand{\phi}{\varphi}
\renewcommand{\epsilon}{\varepsilon}
\renewcommand{\theta}{\vartheta}
\newcommand{\0}{\emptyset}

\newcommand{\shc}{shc}
\newcommand{\perpcdot}{{\perp}{\kern0.1em{\cdot}\kern0.1em}}

\renewcommand{\split}[1]{\ensuremath{{\mathrm{split}(#1)}}}

\newcommand{\freeze}{\ensuremath{\text{\scriptsize%
\SnowflakeChevron}}}
\newcommand{\sigmacomb}{\ensuremath{\sigma^\freeze_\lc}}
\newcommand{\mb}{{\mathrm{mb}}}

\newcommand{\front}{\mathrm{front}}

 \newcommand{\Pos}{\mathrm{Pos}}
 \newcommand{\Moves}{\mathrm{Mvs}}
 
\newcommand{\Mon}{\mathrm{Mon}}

\newcommand{\cross}{\mathrm{cross}}
\newcommand{\wmon}{\textup{wm}}
\newcommand{\ReachMoves}{\mathrm{ReachMvs}}
\newcommand{\CompMoves}{\mathrm{CompMvs}}
\newcommand{\copmon}{\mathrm{cm}}
\newcommand{\comp}{\mathrm{comp}}
\newcommand{\reach}{\mathrm{reach}}
\newcommand{\robmon}{\mathrm{rm}}

\DeclareMathAlphabet{\mathsc}{OT1}{cmr}{m}{sc}

\newcommand{\nptime}{\ensuremath{\mathsc{NP}}\xspace}

\newcommand{\elorder}{\mathbin{\vartriangleleft}}
\newcommand{\elorderg}{\mathbin{\vartriangleright}}

\newcommand{\elordergeq}{\mathbin{\unrhd}}
\newcommand{\setorder}{\mathbin{\sqsubset}}
\newcommand{\setordereq}{\mathbin{\sqsubseteq}}

\newcommand{\totheright}{\mathbin{\vartriangleleft}}
\newcommand{\prefixordereq}{\mathbin{\sqsubseteq}}

\newcommand{\tw}{tree width\xspace}

\newcommand{\dtw}{directed tree width\xspace}
\newcommand{\Dtw}{Directed tree width\xspace}
\newcommand{\otw}{oriented tree width\xspace}
\newcommand{\Otw}{Oriented tree width\xspace}
\newcommand{\dagw}{DAG-{}width\xspace}
\newcommand{\Dagw}{DAG-{}width\xspace}
\newcommand{\kw}{Kelly-{}width\xspace}
\newcommand{\Kw}{Kelly-{}width\xspace}
\newcommand{\dw}{D-{}width\xspace}

\newcommand{\wm}{{\mathrm{wm}}}

\newcommand{\shysim}{{\mathrm{shy}\text{-}\mathrm{sim}}}
\newcommand{\lc}{{\mathrm{lc}}} 

\DeclareMathOperator{\dTW}{dtw}

\newcommand{\dTWcops}[1]{cn_{#1}(\dTW)}
\newcommand{\cmdTWcops}[1]{cn_{#1}\mathrm{(cmdtw)}}
\newcommand{\rmdTWcops}[1]{cn_{#1}\mathrm{(rmdtw)}}
\newcommand{\shyDAGWcops}[1]{cn_{#1}\mathrm{(shyDAG)}}
\newcommand{\wmDAGWcops}[1]{cn_{#1}\mathrm{(wmDAG)}}
\newcommand{\wmshyDAGWcops}[1]{cn_{#1}\mathrm{(wmshyDAG)}}
\newcommand{\DAGWcops}[1]{cn_{#1}\mathrm{(DAG)}}
\DeclareMathOperator{\wmDAGW}{wm\text{-}\mathrm{DAG}\text{-}\mathrm{w}}

\DeclareMathOperator{\DAGW}{DAG\text{-}w}
\DeclareMathOperator{\KW}{Kelly\text{-}w}

\DeclareMathOperator{\DW}{D\text{-}w}
\DeclareMathOperator{\DS}{DS\text{-}w}

\DeclareMathOperator{\oTW}{otw}

\newcommand{\ie}{i.e.\@\xspace}
\newcommand{\Ie}{I.e.\@\xspace}

 \usepackage[blocks,auth-lg]{authblk}

\begin{document}

\title{Directed Width Measures and Monotonicity of Directed Graph Searching}

\renewcommand{\Affilfont}{\itshape}

\author{\L ukasz~Kaiser}
\affil{CNRS \& Universit\'e Sorbonne Paris Cit\'e\footnote{Currently
  at Google Inc.}, \authorcr\texttt{lukaszkaiser$@$google.com}}
\author{Stephan~Kreutzer}
\author{Roman~Rabinovich\thanks{This work was partially supported by
    the projects \textit{Games for Analysis and Synthesis of Interactive
Computational Systems (GASICS)} and \textit{Logic for Interaction (LINT)} of the
\textit{European Science Foundation}.}}
\author{Sebastian~Siebertz}
\affil{Logic and Semantics, Technical University Berlin \authorcr
$\{$\texttt{%
stephan.kreutzer,roman.rabinovich,
sebastian.siebertz}$\}$%
$@$\texttt{tu-berlin.de}}


\maketitle

\begin{abstract}
  We consider generalisations of \tw to directed graphs, that
  attracted much attention in the last fifteen years. About their
  relative strength with respect to ``bounded width in one measure
  implies bounded width in the other'' many problems remain
  unsolved. Only some results separating directed width measures are
  known. We give an almost complete picture of this relation.

  For this, we consider the cops and robber games characterising
  DAG-width and directed tree width (up to a constant factor). For
  DAG-width games, it is an open question whether the
  robber-monotonicity cost (the difference between the minimal numbers
  of cops capturing the robber in the general and in the monotone
  case) can be bounded by any function. Examples show that this
  function (if it exists) is at least $f(k) > 4k/3$
  \cite{KreutzerOrd08}. We approach a solution by defining \emph{weak
  monotonicity} and showing that if~$k$ cops win weakly monotonically,
  then $O(k^2)$ cops win monotonically. It follows that bounded
  Kelly-width implies bounded DAG-width, which has been open since the
  definition of Kelly-width \cite{HunterKre08}.

  For directed tree width games we show that,
  unexpectedly, the cop-monotonicity cost (no cop revisits any vertex)
  is not bounded by any function. This separates directed tree width
  from D-width defined in~\cite{Safari05}, refuting a conjecture 
  in~\cite{Safari05}.
\end{abstract}

\section{Introduction}

     In the study of hard algorithmic problems on graphs,
methods derived from structural graph theory have proved to be a
valuable tool. The rich theory of special classes of graphs developed
in this area has been used to identify classes of graphs, such as
classes of bounded tree width or clique-width,
on which many computationally hard problems can be solved
efficiently. 
Most of these classes are defined by some structural
property, such as having a tree decomposition of low width, and this structural 
information can be exploited
algorithmically. 

Structural parameters such as tree width, clique-width, classes of graphs
defined by excluded minors etc.\@ studied in this context relate to
undirected graphs.
However, in various applications in computer science, directed graphs are
a more natural model. 
Given the enormous success
width parameters 
had for problems defined on undirected graphs, it is natural to ask
whether they can also be used to analyse the complexity of hard
algorithmic problems
on digraphs.
While in principle it is 
possible to apply the structure theory for undirected graphs to directed
graphs by ignoring the direction of edges, this implies a significant
information loss. Hence,
for computational problems whose instances are directed graphs,  
methods based on the structure theory for undirected graphs may be
less useful.

Reed \cite{Reed99} and Johnson, Robertson, Seymour and
Thomas~\cite{JohnsonRobSeyTho01} initiated the development of a
decomposition theory for directed graphs with the aim of defining a
directed analogue of undirected tree width. They introduced the
concept of \emph{directed tree width} and showed that the $k$-disjoint
paths problem and more general linkage problems can be solved in
polynomial-time on classes of digraphs of bounded directed
tree width. Following this initial proposal, several alternative
notions of width measures for sparse classes of digraphs have been
introduced, for instance \emph{directed path width} (see
\cite{Barat06}, initially proposed by Robertson, Seymour and Thomas),
\emph{D-width}~\cite{Safari05},
\emph{DAG-width}~\cite{BerwangerDawHunKreObd12} and
\emph{Kelly-width}~\cite{HunterKre08}. For each of these, algorithmic
applications were given, for example in relation to linkage problems
or a form of combinatorial games known as \emph{parity games}. On the
other hand, some other standard graph theoretical problems such as
directed dominating set remain intractable on classes of digraphs of
small width with respect to these measures. More recently, directed
width parameters have been used successfully in areas outside core
graph algorithmics, for instance in Boolean network
analysis~\cite{Tamaki10}, in the evaluation of simple regular path
queries~\cite{BaganBG13}, in the theory of verification in form of
$\mu$-calculus model-checking and solving parity
games~\cite{BerwangerDawHunKreObd12, HunterKre08, BerwangerGra05}.

Despite the considerable interest these parameters have generated, not
much is known about the relation between them. It is known that
classes of bounded DAG-width, Kelly-width or D-width also have bounded
directed tree width, making directed tree width the most general of
these parameters. On the other hand, classes of digraphs of bounded
directed path width also have bounded width in the other
measures. However, it is still an open problem how DAG-width,
Kelly-width or D-width relate to each other. The main structural
contribution of this paper is to give an almost complete picture of
the relationship between these width parameters with strict
inequalities in most cases.

Digraph parameters such as directed tree width, DAG-width or
Kelly-width are closely related to graph searching games, also called
cops and robber games in this case. In a graph searching game, a
number of cops tries to capture a robber on a graph or digraph. The
robber occupies a vertex of the graph and so does each of the
cops. The game is played in rounds where in each round the cops first
announce their new position and then the robber can move to a
different vertex of the graph to avoid capture. See below for details
and see~\cite{FominThi08,Kreutzer11b} for recent surveys.

Variations of the game are obtained by restricting the moves of the
cops and the robber in several ways.  On every graph or digraph, the
cops have a winning strategy that guarantees capturing the robber by
using sufficiently many cops. The minimal number of cops on a
digraph~$G$ that guarantees to capture the robber is a natural
graph parameter and it turns out that the width measures discussed
above are closely related to these parameters defined by suitable
graph searching games.

An important concept in the context of graph searching games is
\emph{monotonicity}.  Monotonicity is a restriction on the strategies
employed by the cops. We distinguish between cop- and robber-monotone
cop strategies. Roughly speaking, a strategy is \emph{cop-monotone} if
the cops never revisit a vertex where they have been before, and it is
\emph{robber-monotone} if the set of vertices that the robber can
occupy never increases during a play. Usually, monotone variants of
graph searching games yield nice decompositions corresponding to
directed or undirected width measures. For instance, a
tree decomposition corresponds exactly to a cop-monotone winning
strategy for the cops in a particular type of graph searching games.

The \emph{(cop- or robber-) monotonicity problem} for variants
of graph searching games---\ie the problem whether on every graph or
digraph the number of cops required to capture a robber with a cop- or
robber-monotone strategy is the same as the number of cops required
with an unrestricted strategy---has intensively been studied in the
literature.  For games that are not monotone, we call the number of
extra cops required for a monotone strategy the \emph{monotonicity
  cost} of the game variant.  For graph searching games on undirected
graphs this problem has been solved for most commonly used game
variants and usually the games are monotone. For directed graphs,
however, the situation is much less understood. It was shown in
\cite{JohnsonRobSeyTho01,Adler07} that the games corresponding to
directed tree width are not monotone. In~\cite{KreutzerOrd08} it was
shown that also the games corresponding to Kelly- and DAG-width are
non-monotone. More precisely, in~\cite{KreutzerOrd08} examples are
exhibited where monotone strategies require at least $\frac 43k$ cops,
but $k$ cops suffice for an unrestricted strategy. However, all
attempts to use the tricks facilitated in these examples to show that
the monotonicity cost is in fact unbounded have failed so far.

Among the most important open problems in the area of cops and robber
games at the moment is the question whether the monotonicity cost for
the games corresponding to directed tree width, Kelly-width or
DAG-width can be bounded by a constant factor, or by any function at
all. This question is particularly interesting for DAG-width and
Kelly-width games, as it was shown in~\cite{HunterKre08} that bounding
the monotonicity cost of these games would imply that DAG-width and
Kelly-width are bounded by each other, \ie a class of digraphs has
bounded Kelly-width if, and only if, it has bounded DAG-width. The
proof relies on translating monotone strategies in one type of game
into (non-monotone) strategies of the other type of game.

For directed tree width games and robber-monotone strategies, the
monotonicity question was answered in the affirmative
in~\cite{JohnsonRobSeyTho01}.  It has been conjectured (\cite[Page
750]{Safari05}\footnote{Safari actually conjectures that D-width
  equals directed tree width which would imply cop-monotonicity.})
that the cop-monotonicity cost should also be bounded for directed
tree width games.  Whether the monotonicity cost for DAG- and
Kelly-width games is bounded is still open as well, despite
considerable efforts in the community. These monotonicity problems are
arguably the most important open problems in cops and robber games.

In this paper we give a negative answer to the cop-monotonicity
problem for directed tree width games. We show that there is a class
of digraphs where $4$ cops have a winning strategy in the directed
tree width game, but the number of cops required to win with a
cop-monotone strategy is unbounded.  We also make progress on the
problem for DAG-width games. We introduce a weaker form of
monotonicity, called \emph{weak monotonicity}, and show that any
weakly monotone strategy for $k$ cops can be transformed into a
robber-monotone strategy for $k^2$ cops. While this does not settle
the monotonicity problem for DAG-width games completely, it
constitutes significant progress towards this longstanding open
problem is the following sense: in the known examples for
non-monotonicity of DAG-width games, for instance
in~\cite{KreutzerOrd08}, the (unrestricted) strategies used by cops to
win the game are actually weakly monotone in our sense. Hence, our
result implies that these tricks cannot be used to show that there is
no bound on the monotonicity cost for DAG-width games. Furthermore, as
explained above, in~\cite{HunterKre08} it is shown that monotone
strategies in the DAG-width or Kelly-width game can be translated into
(non-monotone) strategies in the other type of games (with roughly the
same number of cops). It turns out that the translation from
Kelly-width games into DAG-width games actually translates a
Kelly-strategy into a weakly monotone DAG-strategy and hence, by our
result, this strategy can further be translated into a monotone
strategy (with a quadratic number of cops). As a consequence, bounded
Kelly-width implies bounded DAG-width, settling one of the open
problems in the relation between different width measures. Finally, a
winning cop strategy in weakly monotone DAG-width game induces a
decomposition of the graph of small width, similar to \tw, \dagw
etc. In contrast to DAG decompositions, for which we do not know
whether there exist ``small'' decompositions (\ie of size polynomial both
in $|G|$ and in the \dagw of~$G$), the new decompositions are essentially
tress (rather than DAGs) of size in $O(|G|^2)$. Having a simpler
structure than DAG decompositions they may be interesting by
themselves both for algorithmical applications and for theoretical
research on DAG-width. We remark that such a decomposition encodes in a
compact way a DAG decomposition of width at most quadratically larger
than the optimal one.

\smallskip

\noindent\textbf{Organisation. } The paper is organised as follows. In
Section~\ref{sec:prelims}, after fixing some basic notation, we
introduce graph searching games and prove our first main result, that
the cop-monotonicity cost for directed tree width games are unbounded
(Theorem~\ref{thm_non_cop_mon}). As a consequence, we separate
directed tree width from D-width.  Our monotonicity results for
DAG-width are presented in Section~\ref{sec:variants} (see
Theorem~\ref{thm_main}). In Section~\ref{sec:ordering}, we compare the
various directed width measures with respect to the question whether
classes of digraphs of bounded width in one measure have bounded width
in another measure. \RR{Stephan, kannst du hier oder früher otw erwähnen?}


\section{Preliminaries}\label{sec:prelims}

     We assume familiarity with basic concepts of directed graph theory and
refer to~\cite{Diestel12} for background. The first part of this section serves
to review and fix notation and terminology. 

We denote the set of positive integers by $\mathbb{N}$ and for
$n\in\mathbb{N}$ we write $[n]$ for the set $\{1,\ldots, n\}$. The
prefix relation on words over some alphabet $\Sigma$ is denoted
by~$\prefixordereq$ and its irreflexive version by~$\sqsubset$. The
lexicographical order is denoted by~$\preceq$ and its irreflexive
variant by~$\prec$. 
We
write $\Sigma^{\leq n}$ for the set of words over $\Sigma$ of length
at most $n$.

All graphs in this paper are finite, directed and simple, \ie they do
not have loops or multiple edges between the same pair of
vertices. Undirected graphs are directed graphs with symmetric edge
relation, but we write $\{v,w\}$ for the undirected edge between~$v$
and~$w$, \ie for the pair of edges $(v,w)$ and $(w,v)$. If~$G$ is a
graph, then $V(G)$ denotes its set of vertices and $E(G)$ its set of
edges. For a set $X\subseteq V(G)$ we write $G[X]$ for the subgraph
of~$G$ induced by~$X$ and $G-X$ for $G[V(G)\setminus X]$. The set of
vertices reachable from a vertex $v\in V(G)$ in~$G$ is denoted by
$\Reach_G(v)$. For a set $X\subseteq V(G)$ we write $\Reach_G(X)$ for
the set $\{w\in V(G) \mid \text{ there is }v\in X \text{ such that
}w\in\Reach_G(v)\}$. For vertices $v,w\in V(G)$ we write $v\ge w$ (or
$w\le v$) if $w\in \Reach_G(v)$ and $v>w$ (or $w<v$) if, additionally,
$v\neq w$. A path is a sequence of vertices $v_1,v_2,\ldots$ with
$(v_i,v_{i+1})\in E(G)$ for all $i\ge 0$. A \emph{strongly connected
  component} of a digraph~$G$ is a maximal subgraph~$C$ of~$G$ which
is strongly connected, \ie~for all $u,v\in V(C)$ we have
$u\in\Reach_C(v)$ and $v\in\Reach_C(u)$. All components considered in
this paper will be strong and hence we simply write \emph{component}.
We write $e\sim v$ if vertex~$v$ is incident with edge~$e$.

We write~$\bar G$ for the underlying undirected graph of~$G$. The
depth of an undirected, rooted tree is the the maximum number of
vertices on a path from the root to a leaf of the tree. We write
$T^d_\ell$ for the complete undirected tree of branching degree~$d$
and depth~$\ell$. We assume that the vertices of a rooted undirected
tree of maximum branching degree~$d$ are words over $\{0,\ldots,
d-1\}$, hence the vertex set of $T^d_\ell$ is $\{0,\ldots, d-1\}^{\leq
\ell}$.

\paragraph{Graph searching games}
A graph searching game is played on a graph $G$ by a team of cops and a 
robber. The robber and each cop occupy a vertex of $G$. Hence, a
current game position can be described by a pair $(C, v)$, where $C$
is the set of vertices occupied by cops and $v$ is the current robber
position. At the beginning the robber chooses an 
arbitrary vertex $v$ and the game starts at position
$(\emptyset, v)$. The game is played
in rounds. In each round, from a position $(C, v)$ the cops first
announce their next move, 
\ie~the set $C'\subseteq V(G)$ of vertices that they will 
occupy next. Based on the triple $(C, C', v)$ the robber chooses his new vertex 
$v'$. This
completes a round and the play continues at position $(C', v')$. 
Variations of graph
searching games are obtained by restricting the moves allowed for the
cops and the robber. In all game variants considered here, from a 
position $(C, C', v)$, \ie~where the cops move from their current
position $C$ to $C'$ and the robber is on $v$, the robber would have
exactly the same choice of possible moves
from any vertex in the component of $G-C$ containing $v$.
 We will therefore describe game positions by a pair $(C, R)$,
or a triple $(C, C', R)$,
where $C, C'$ are as before and~$R$ is a component of $G-C$. We
call~$R$ the \emph{robber component.} 


Formally a graph searching game on a graph $G$ is specified by a tuple 
$\calG = (\Pos(G)$, $\Moves(G), \Mon)$, where $\Pos(G)$ describes the set of
possible positions, $\Moves(G)$ the set of legal moves and $\Mon$ specifies
the monotonicity criteria used. In all game variants considered
here, the set $\Pos(G)$ of positions is $\Pos_c \cup \Pos_r$ where $\Pos_c =  
\{(C,R) \mid  C\subseteq V(G)$ ,
$R\subseteq V(G)$ is a component of $G- C\}$ are cop positions and $\Pos_r =
\{(C,C',R) \mid C, C' \subseteq V(G)$ and $R\subseteq  
V(G)$ is a component of $G-C\}$ are robber positions.

As far as legal moves are concerned, we distinguish between two
different types of games, called \emph{reachability} and
\emph{component} games. In both cases the  
cops moves are $\Moves_c(G) := \{\bigl( (C,R), (C,C',R) \bigr) \mid (C,R)\in 
\Pos_c, 
(C,C',R)\in \Pos_r\}$. The difference is in the
definition of the set of possible robber moves.

\paragraph*{Reachability game}  
In the
\emph{reachability game}, $\Moves(G)$ are defined as $\ReachMoves(G)$, where $\ReachMoves(G)$ $:=  \Moves_c(G) \cup 
\{\bigl(
(C,C',R),(C',R') \bigr) \mid
 (C,C',R)\in \Pos_r$, $(C',R')\in \Pos_c$ and $R'$ is a component of
 $G-C'$ reachable from a vertex in $R$ by a directed path in $G-(C\cap
 C') \}$. \Ie~the robber can run along any directed path in the
 digraph which does not contain a cop from $C\cap C'$ (i.e.~one that remains on
 the board). 

\paragraph*{Component game}  In the \emph{component
  game}, we define $\Moves(G)$ as $\CompMoves(G)$,
where $\CompMoves(G)$ $:=  \Moves_c(G) \cup  \{\bigl(
(C,C',R),(C',R') \bigr) \mid
 (C,C',R)\in \Pos_r$, $(C',R')\in \Pos_c$ and $R'$ is a component of
 $G-C'$ such that $R, R'$ are contained in the same component of
 $G-(C\cap C') \}$. \Ie in the component game, the robber can
 only run to a new vertex within the strongly connected component of
 $G-(C\cap C')$ that contains his current position.

\smallskip

\paragraph*{Monotonicity} The component $\Mon$ describes the
monotonicity condition and is a set of finite plays. The cops win all
plays $(C_0,R_0), (C_0,C_1,R_0),(C_1,R_1),\ldots$ in $\Mon$ where $R_i
= \emptyset$ for some~$i$ and the robber wins all other plays.
Usually $\Mon$ describes cop- or robber-monotonicity; the latter is
defined differently in the component game and in the reachability
game: $\Mon\subseteq \copmon(G)\cup \robmon_\comp(G)\cup
\robmon_\reach(G)$.  A play $(C_0,R_0), (C_0,C_1,R_0),$
$(C_1,R_1),\ldots$ is in $\copmon_\reach(G)$, called
\emph{cop-monotone,} if for all $i,j,k\ge 0$ with $i<j<k$ we have
$C_i\cap C_k \subseteq C_j$, \ie cop-monotonicity means that the cops
never reoccupy vertices.  A play $(C_0,R_0),
(C_0,C_1,R_0),(C_1,R_1),\ldots$ is in $\robmon_{\reach}(G)$, called
\emph{robber-monotone,} if for all~$i$, $\Reach_{G-(C_i\cap
  C_{i+1})}(R_{i+1})\subseteq \Reach_{G-(C_i\cap C_{i+1})} (R_i)$,
\ie~once the robber cannot reach a vertex, he won't be able to reach
it forever. Finally, a play is in $\robmon_{\comp}(G)$, also called
\emph{robber-monotone}, if $R_{i+1}\subseteq R_i$ for all~$i$.

A strategy for the cops is \emph{cop-} or \emph{robber-monotone} if all plays 
consistent with that strategy are
cop- or robber-monotone, respectively.

By combining reachability or component games with monotonicity
conditions we obtain a range of different graph searching games. 
It follows immediately from the definition that on every digraph the
cops have a winning strategy in each of the graph searching games
defined above by simply placing a cop on every vertex. For a given
digraph $G$, we are
therefore interested in the minimal number $k$ such that the cops have a
winning strategy in which no cop position $C_i$ contains more than $k$
vertices. 

\begin{definition}\label{def:game-defs}
  For any digraph $G$, we define
  \begin{itemize}[noitemsep]
  \item $\dTWcops{G}$ as the minimal number of cops needed to win
    $(\Pos(G),$ $\CompMoves(G),$ $\Mon)$ where $\Mon$ is the set of
    all finite plays,
  \item $\cmdTWcops{G}$ as the minimal number of cops to win
    $(\Pos(G),$ $\CompMoves(G),$ $\Mon = \copmon_\comp(G))$,
  \item $\DAGWcops{G}$ as the minimal number of cops needed to win
    $(\Pos(G),$ $\ReachMoves(G),$ $\Mon = \copmon_\reach(G))$.
  \end{itemize}
\end{definition}

It follows immediately from the definitions that, for all digraphs $G$,
$\dTWcops{G}\leq \DAGWcops{G}$ and $\dTWcops{G}\leq \cmdTWcops{G}$.
The number $\cmdTWcops{G}-\dTWcops{G}$ is called the
\emph{cop-{}monotonicity cost} for the component game on~$G$.
Robber-monotonicity cost as well as monotonicity cost for other game
variants are defined analogously.


\section{Strong non-cop-monotonicity of \dtw}

     \Dtw can be characterised up to a constant factor by the \dtw game. 

\begin{theorem}[\cite{JohnsonRobSeyTho01}]\label{thm:dtw_cndtw}
The \dtw of a graph $G$ and $\dTWcops{G}$ are within a constant factor 
of each other.
\end{theorem}
\Dtw is defined by directed tree decompositions
(in~\cite{JohnsonRobSeyTho01} called \emph{arboreal decompositions}),
see \Cref{subsec:otw} for a definition. Such a decomposition can be
viewed as a description of a robber-monotone winning strategy for the
cops. The proof of Theorem~\ref{thm:dtw_cndtw} essentially shows that
a winning strategy for~$k$ cops can be transfered in a directed tree
decomposition of width, roughly, at most~$3k$ and hence in a
robber-monotone winning strategy for approximately $3k$ cops.
It follows that the robber-monotonicity cost for \dtw is
bounded by a constant factor. 

One would expect that the cop-monotonicity cost can be bounded
similarly by a slowly growing function. However, the following theorem
shows that the cop-monotonicity cost for \dtw cannot be bounded by any
function at all.

\begin{theorem}\label{thm_non_cop_mon}
  There is a class $\{G_n \mid n > 2\}$ of graphs such that for all
  $n$, $\dTWcops{G_n} \le 4$ and $\cmdTWcops{G_n}\ge n$.
\end{theorem}
\begin{proof}
  Let $n>2$. We inductively define a sequence of graphs $G_n^m$ and
  sets marked vertices $M(G_n^m) \subseteq V(G_n^m)$ for
  $m\in\{1,\ldots, n+1\}$. We then define $G_n$ as $G_n^{n+1}$.

  First $G_n^1$ is an edgeless graph with a single vertex and
  $M(G_n^1)=V(G_n^1)$, \ie the vertex of $G_n^1$ is marked. Assume
  that $(G_n^m, M(G_n^m))$ has been constructed. Recall that
  $T^d_\ell$ denotes a complete undirected tree of branching degree
  $d$ and depth~$\ell$.  One part of $G_n^{m+1}$ is a copy of
  $T_{n+2}^{n+1}$, which has $(n+1)^{n+2}$ leaves $v_s$ for $s \in
  \{1,\ldots,(n+1)^{n+2}\}$. The graph $G_n^{m+1}$ is the disjoint
  union of $T_{n+2}^{n+1}$ and $n\cdot(n+1)^{n+2}$ copies
  $H^{m+1}_j(v_s)$ of $G_n^m$ where $j\in\{1,\ldots,n\}$ and
  $s\in\{1,\ldots,(n+1)^{n+2}\}$ plus some additional edges which we
  describe next. We denote the subgraph $T_{n+2}^{n+1}$ of $G_n^{m+1}$
  by $T(G_n^{m+1})$.


  For every leaf $v \in\{v_s \mid 1\le s\le (n+1)^{n+2}\}$ there is an
  undirected edge from~$v$ to the root of $H^{m+1}_i(v)$. Let
  $x^{m+1}_i(v)$ be the $i$th vertex on the path from the root of
  $T(G_n^{m+1})$ to~$v$. For all leaves~$v$ of $T(G_n^{m+1})$ and all
  $1\leq i\leq n$ we add directed edges from $x^{m+1}_i(v)$ to all
  marked vertices $M(H^{m+1}_i(v))$ of $H^{m+1}_i(v)$. Finally, for
  all leaves $v$ of $T(G_n^{m+1})$ and all leaves of $H_i^{m+1}(v)$ we
  add a directed edge to $v$.  We define $M(G_n^{m+1})\coloneqq
  V(T(G_n^{m+1}))$. The graph $G_n$ is schematically shown in
  \Cref{fig_noncopmon} (edges without arrows mean edges in both directions). 
\begin{figure}
\begin{center}

        \newcommand{\triangdots}[3]{%
          \node[vertex] (epsilon#3) at ($ (#1,#2) + (0,0) $){};
          \node[vertex] (0) at ($ (#1,#2) + (-1,-1) $){};
          \node at ($ (#1,#2) + (0,-1) $){$\cdots$};
          \node[vertex] (1) at ($ (#1,#2) + (1,-1) $){};
          \draw (epsilon#3) edge (0); \draw (epsilon#3) edge (1);
          \node[rotate=10]  at ($ (#1,#2) + (-1.75,-1.5) $) {$\iddots$};
          \node[rotate=-10] at ($ (#1,#2) + (-0.25,-1.5) $) {${}\ddots{}$};
          \node[rotate=0] at ($ (#1,#2) + (1,-1.75) $) {${}\vdots{}$};
        }
        
        \newcommand{\triangline}[4]{%
          \node[vertex] (epsilon#3) at ($ (#1,#2) + (0,0) $){};
          \node[vertex] (0#4) at ($ (#1,#2) + (-1,-1) $){};
          \node at ($ (#1,#2) + (0,-1) $){$\cdots$};
          \node[vertex] (1) at ($ (#1,#2) + (1,-1) $){};
          \draw (epsilon#3) edge (0#4); \draw (epsilon#3) edge (1);
          \draw (0#4) -- ++(-1,-1);
          \node[rotate=0] at ($ (#1,#2) + (-0.25,-1.5) $) {${}\ddots{}$};
          \node[rotate=0] at ($ (#1,#2) + (1,-1.75) $) {${}\vdots{}$};
          \node[rotate=0] at ($ (#1,#2) + (-1,-2) $) {${}\cdots{}$};
        }
        
        \newcommand{\triang}[2]{%
          \node[vertex] (epsilon) at ($ (#1,#2) + (0,0) $){};
          \node[vertex] (0) at ($ (#1,#2) + (-1,-1) $){};
          \node at ($ (#1,#2) + (0,-1) $){$\cdots$};
          \node[vertex] (1) at ($ (#1,#2) + (1,-1) $){};
          \draw (epsilon) edge (0); \draw (epsilon) edge (1);
        }
        
\begin{tikzpicture}[scale=0.5]

        \triangdots{0.5}{-4.5}{7}
        \node[draw,rectangle,rounded corners,minimum width=2.9cm,minimum
        height=0.5cm,rotate=45] (rectangle0) at (-1.2,-6.25){};
        
        \triangline{-1.9}{-7}{8}{7}
        \triangdots{-4}{-9.1}{9}
        
        \triangdots{0}{0}{0}
        \draw[-slim,blue,bend left=40,dashed] (0.east) to (rectangle0);
        
        \triangline{-2.4}{-2.4}{1}{1}
        \draw (01) to (epsilon7);
        
        \triangdots{-4.5}{-4.5}{2}
        \triangline{-7}{-7}{3}{3}
        \triangdots{-9.1}{-9.1}{4}
        \triangline{-11.5}{-11.5}{5}{5}
        \triangdots{-13.6}{-13.6}{6}
        \triang{-16}{-16}
        
        \node[vertex] (leaf0) at (-18,-18){};
        \node[vertex] (leaf1) at (-16,-18){};
        \draw (leaf0) to (0); \draw (leaf1) to (0);
        
        \node[draw,rectangle,rounded corners,minimum width=2.8cm,minimum
        height=0.5cm,rotate=45] (rectangle1) at (-1.7,-1.7){};
        \node[draw,rectangle,rounded corners,minimum width=2.9cm,minimum
        height=0.5cm,rotate=45] (rectangle2) at (-6.25,-6.25){};
        \node[draw,rectangle,rounded corners,minimum width=2.8cm,minimum
        height=0.5cm,rotate=45] (rectangle3) at (-10.8,-10.8){};
        \node[draw,rectangle,rounded corners,minimum width=2.8cm,minimum
        height=0.5cm,rotate=45] (rectangle4) at (-15.3,-15.3){};
        
        \draw[-slim,red,bend left=40,dotted,thick] (leaf0) to (01);
        \draw[-slim,red,bend left=40,dotted,thick] (leaf0) to (03);
        \draw[-slim,red,bend left=40,dotted,thick] (leaf0) to (05);
        
        \draw[-slim,blue,bend right=40,dashed] (epsilon0) to (rectangle2);
        \draw[-slim,blue,bend right=40,dashed] (epsilon2) to (rectangle3);
        \draw[-slim,blue,bend right=40,dashed] (epsilon4) to (rectangle4);


         \node[] at (-4,-17) {$G_n = G^{n+1}_n$};
         \node[] at (-0,-11) {$H^{n+1}_2(0^n)$};
         \node[] at (-8,-14) {$H^{n+1}_1(0^n)$};
         \node (0^n) at (-8,-1) {$v=0^n$};
         \draw[ie] (0^n) -- (01);

\end{tikzpicture}
\end{center}
\caption{$\dTWcops{G_n}=4$, but 
      the robber wins against $n$ cop-{}monotone cops. Only the left-most
      branch of $G_n$ and the upper part of the left-most branch of $G_2^n(0^n)$ is shown.}
 \label{fig_noncopmon}
 \end{figure}
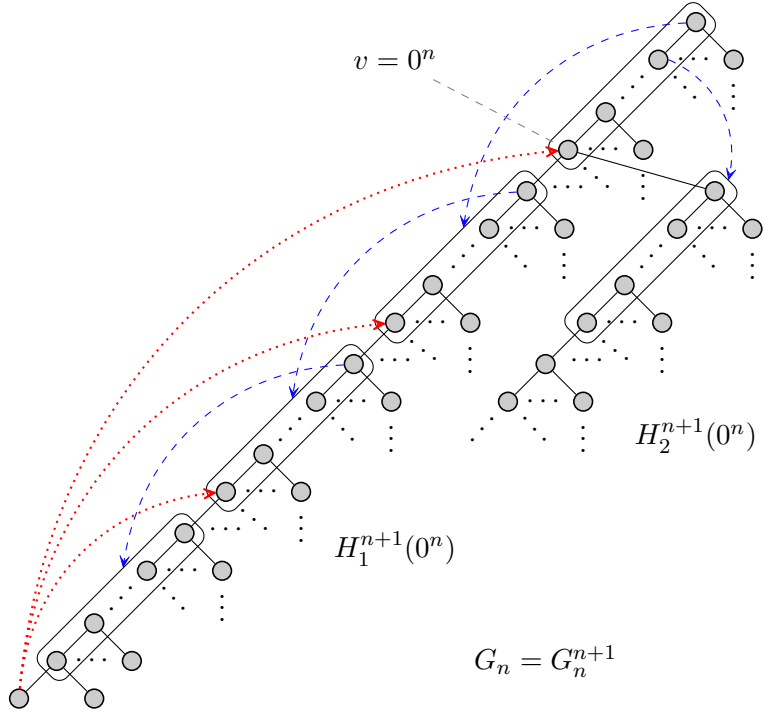
  
  \smallskip Let us describe a non-{}cop-{}monotone winning strategy
  for~$4$ cops on~$G_n$.  Observe that $G_n=G_n^{n+1}$ is an
  undirected tree with additional edges that connect only vertices of
  the same branch.  In particular, for each subgraph $H_j^i(v)$, if
  the robber is in $H_j^i(v)$ and the cops block the root of
  $T(H_j^i(v))$ and $x_j^{i+1}(v)$, then the robber may not leave
  $H_j^i(v)$ as he cannot reenter $H_j^i(v)$. 

  We show that in each play of the game there is a unique sequence
  \[G_n^{n+1}, H_{j(n)}^{n}(v_n), H_{j(n-1)}^{n-1}(v_{n-1}),\ldots,
  H_{j(1)}^{1}(v_1)\] of subgraphs in which the cops are placed and
  such that the robber is captured on the unique vertex of
  $H_{j(1)}^1(v_1)$.

  Assume that the root of $T(G_n^{n+1})$ is occupied by a cop. Then
  two additional cops can play in a top-{}down manner in
  $T(G_n^{n+1})$ following the robber to his tree branch until the
  robber is forced out of $T(G_n^{n+1})$ into some $H_{j(n)}^{n}(v)$
  for some leaf $v$ of $T(G_n^{n+1})$. Define $v_n\coloneqq v$. The
  cops now occupy in a first step $v_n$, the root of
  $T(H_{j(n)}^{n}(v_n))$ and $x_{j(n)}^{i+1}(v_n)$. In a second step,
  they release the cop from $v_n$ and from the root of $T(G_n^{n+1})$,
  as these vertices are no longer available for the robber.

  Similarly, assume that the root of $T(H_{j(i)}^{n}(v_i))$ and
  $x_{j(i)}^{i+1}(v_i)$ are occupied by cops. As above, two additional
  cops can play in a top-{}down manner in $T(H_{j(i)}^{n}(v_i))$
  following the robber to his tree branch until the robber is forced
  out of $T(H_{j(i)}^{n}(v_i))$ into some $H_{j(i-1)}^{i-1}(v)$ for
  some leaf $v$ of $T(H_{j(i)}^{n}(v_i))$. Define $v_{i-1}\coloneqq
  v$. At this point of time, three cops are placed on the graph, one
  on $v_{i-1}$, one on the root of $T(H_{j(i)}^{n}(v_i))$ and one on
  $x_{j(i)}^{i+1}(v_i)$. The cops now first occupy with an additional
  cop the root of $T(H_{j(i-1)}^{n}(v_{i-1}))$. They can now release
  the cop from $v_{i-1}$ which they place on
  $x_{j(i-1)}^{i}(v_{i-1})$. Finally they may release the cop from
  $x_{j(i)}^{i+1}(v_i)$ and thereby establish the induction hypothesis
  for~$i-1$.

  In this way the robber is captured at the latest on the single vertex of
  $H_{j(1)}^{1}(v_1)$.



\smallskip
Now we construct a robber strategy that wins against all cop-{}monotone 
strategies for~$n$ cops if $n>2$. For a vertex~$v$ and subtree~$T$ 
of~$G_n$ we say that~$T$ is a subtree of~$v$ if the root of~$T$ is a direct 
successor of~$v$. The robber resides on a vertex of $T(G_n)$ that has the least 
distance to the root of $G_n$ as long as this is possible. When a cop occupies 
his vertex~$v$ the robber proceeds to a directed successor of~$v$ such that the 
subtree of~$v$ is cop free. Such a successor always exists due to the 
high branching degree of~$G_n$. When the robber reaches a leaf~$w_n$ of 
$T(G_n)$, every vertex on the path from the root of $G_n$ to~$w_n$ has been 
occupied by a cop. As the length of the path is greater that the number of 
cops, there is a vertex $x^n_{i_n}(w_n)$ that has been left by a cop. When a 
cop occupies~$w_n$, the robber goes to $G^n_{i_n}(w_n)$. Now on 
$G^n_{i_n}(w_n)$ (which is isomorphic to $G_n^{n-1}$) the robber plays in the 
same way as on~$G_n$ and so on recursively for each~$m$ on $G^m_{i_m}(w_m)$. 
Note that until the robber is captured, there is a path from this vertex to a 
leaf of $G_n$ and then to all already chosen~$w_j$. 

Consider a position when the robber arrives at a leaf~$v$ of~$G_n$ and a cop is 
landing on this vertex. Then at most~$n-1$ cops are on the graph and there is 
some~$j$ such that there is no cop in~$T(G^j_{i_j}(w_j))$. Thus there is a cop 
free path from~$v$ to~$w_j$, then to $x^j_{i_j}(w_j)$ within $T(G^j_{i_j}(w_j))$ 
and then via $x^{j-1}_{i_{j-1}}(w_{j-1})$, $x^{j-2}_{i_{j-2}}(w_{j-2})$, 
\ldots, $x^{2}_{i_{2}}(w_{2})$ back to~$v$. Note that all those $x$-vertices are 
not 
occupied by cops by construction. Thus the robber can return to $w_j$ and play 
from~$w_j$ as before. In this way the robber will never be captured.
\end{proof}


\section{Towards monotonicity of the DAG-width game}\label{sec:variants}

    \dagw is usually defined by means of DAG decompositions, similar to
tree decompositions. For our purposes a game theoretic
characterisation of the \dagw of a graph~$G$ as $\DAGWcops{G}$ is more
useful, and we take it as a definition and refer to the corresponding
game as \dagw game. See~\cite{BerwangerDawHunKreObd12} for details.

As explained in the introduction, one of the most important open
problems in graph searching is the question whether cop- and
robber-monotonicity cost of \dagw games is bounded by any
function. Towards this goal, we introduce two new constraints for the
\dagw game, \emph{weak monotonicity} and a technical notion of \emph{shyness}.

\emph{Weak monotonicity} relaxes the winning condition for the cops,
so that they win more plays.  For a digraph $G$ we define
$\wmon(G)$ as the set of all finite plays $(C_0,R_0),
(C_0,C_1,R_0),(C_1,R_1),\ldots$ such that the following condition is
satisfied. For all $i$ let $c(i) := C_{i+1}\cap R_i$ be the cops which
move into the component of $G - C_i$ currently used by the robber. We
call these cops the \emph{chasers}. All other cops being placed, i.e. the cops 
in
$(C_{i+1}\setminus C_i)\setminus c(i)$ are \emph{guards}. The play $(C_0,R_0),
(C_0,C_1,R_0),(C_1,R_1),\ldots$ is \emph{weakly monotone} if for all~$i$ and
all~$j$ with $j<i$, no vertex in $c(j)$ is reachable by a directed
path from any vertex in~$R_i$ in $G-(C_i\cap C_{i+1})$. That is, for
weak monotonicity we only require monotonicity in the cops that
are used to shrink the robber space but not in the cops placed outside of the
component to block the paths to previous cop positions.
The set $\wmon(G)$ is the set of all weakly monotone plays on $G$.

In a shy robber game, the robber can never leave his strong
component and therefore has the same set of possible moves as in the
directed tree width game. However, and this is crucial, the
monotonicity conditions are defined based on directed
reachability. I.e. the robber can destroy monotonicity if there is a
directed path from his current position to a forbidden vertex. We use
the shy robber games to consider the case in the \dagw game when the robber decides never to
change his component, even if he could do this: we just enforce him to
stay in his component. Of course, this does not restrict his ability
to infur non-monotonicity outside of his component.

Based on weak monotonicity we can now define the following variants of
the \dagw game.
\begin{itemize}[noitemsep]
\item The \emph{weakly monotone game} is the game
  defined by $(\Pos(G),$ $\ReachMoves(G), \Mon =  \wm(G))$.
\item The \emph{weakly monotone shy (robber) game} is the
  game $(\Pos(G),$ $\CompMoves(G), \Mon =  \wm(G) )$.
\item Finally, the \emph{strongly monotone shy (robber) game}
  is the game $(\Pos(G),$ $\CompMoves(G), \Mon =  \copmon_\reach(G) )$.
\end{itemize}

We write $\shyDAGWcops{G}$, $\wmDAGWcops{G}$ and $\wmshyDAGWcops{G}$
for the minimal number~$k$ of cops needed to win the corresponding
game. The following inequalities are immediate consequences of the definitions:
\begin{align*} 
\wmshyDAGWcops{G} \leq \shyDAGWcops{G} \leq \DAGWcops{G}\,;\\
 \wmshyDAGWcops{G} \leq \wmDAGWcops{G} \leq \DAGWcops{G}. 
\end{align*}
The following theorem is our main result in this section.
\begin{theorem}\label{thm_main}
  If $k$ cops capture a shy robber in a weakly monotone way,
  then $18k^2+3k$ cops capture a non-shy robber in a strongly monotone way.
\end{theorem}

Hence the weak monotonicity cost is bounded by a quadratic
function. To prove that the (strong) monotonicity cost is
bounded, it suffices to show that for some function $f\colon
\bbN\to\bbN$ a winning strategy for~$k$ cops in the \dagw game without
any monotonicity constraints induces a winning strategy for $f(k)$
cops against a shy robber in the weakly monotone game. If this is not
true, Theorem~\ref{thm_main} shows that the examples
from~\cite{KreutzerOrd11} cannot be used to prove this, as the
winning strategies used there are weakly monotone.

    \subsection{Blocking and the blocking order}\label{sec_blocking}

Considering the robber whose influence reaches further than his
current component (either because he can leave it or by the weak
monotonicity) we study the properties of cops blocking certain
positions from the robber. This is crucial for placing the guards.
To make this formal, we define blocking sets and an order
on them, so that we can speak about minimal blocking sets.

\begin{definition}[Blocking]
Let $R$, $M$ and $X$ be sets of vertices of a graph $G$. We say
that $X$ \emph{blocks} $R\to M$ if $X\cap R=\emptyset$ and every path
from $R$ to $M$ in $G$ contains a vertex in $X$. When $R$ and $M$
are clear from the context, we simply say that $X$ is a \emph{blocker}.
\end{definition}

Below, we formulate a few basic properties of blocking.
Note that the graph~$G$ can, of course, have cycles.
For some $X\subseteq V(G)$ and a path $P$ we say that $P$ is $X$-free
if $X\cap P = \0$.

\begin{lemma} \label{lemma_transfer-down}
If $X$ blocks $R \to M$ and $Y$ blocks $R \to X$, then~$Y$ blocks $R \to M$.
\end{lemma}
\begin{proof}
  Assume to the contrary that~$Y$ does not block $R \to M$, so there
  is a $Y$-free path~$P$ from~$R$ to~$M$, see Figure
  \ref{fig_lemma_transfer-down}.  Since~$X$ blocks $R \to M$, there is
  a vertex~$v$ on this path which is in~$X$.  But then the prefix up
  to~$v$ of the path~$P$ is a $Y$-free path from~$R$ to~$X$, a
  contradiction to the assumption that~$Y$ blocks $R \to X$.
\end{proof}

\begin{figure}
\begin{center}
\begin{tikzpicture}

\draw (0,0) ellipse (0.2 and 1);
\node at (0,0){$M$};

\draw (1,0) ellipse (0.2 and 1);
\node at (1,0){$X$};

\draw (2,0) ellipse (0.2 and 1);
\node at (2,0){$Y$};

\draw (3,0) ellipse (0.2 and 1);
\node at (3,0){$R$};

\draw[path2=20,-] (3,0.5) -- (2,0.5) -- (1,0.5);\draw[path2=20] (1,0.5)
node[point]{} -- (0,0.5);
\draw[path2=20,-] (3,-0.5) -- (2,-0.5);\draw[path2=20] (2,-0.5) node[point]{} --
(1,-0.5);

\end{tikzpicture}
\caption{Illustration for Lemma \ref{lemma_transfer-down}}
\label{fig_lemma_transfer-down}
\end{center}
\end{figure}
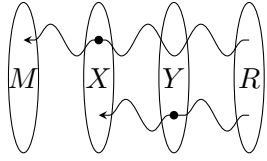

\begin{lemma} \label{transfer-up}
If $X$ blocks $R \to M$ and $Y$ blocks $X \to M$, then~$Y$ blocks $R \to M$.
\end{lemma}

\begin{proof}
Assume to the contrary that $Y$ does not block $R \to M$, so there is
a $Y$-free path~$P$ from $R$ to $M$, see Figure \ref{fig_lemma_transfer-up}.
Since $X$ blocks $R \to M$, there is a vertex $v$ on this path which is in $X$.
But since $Y$ blocks $X \to M$, there must be a vertex in $Y$ on the suffix
of~$P$ starting from $v$. This is a contradiction, as~$P$ was assumed
to be $Y$-free.
\end{proof}

\begin{figure}
\begin{center}
\begin{tikzpicture}

\draw (0,0) ellipse (0.2 and 1);
\node at (0,0){$M$};

\draw (1,0) ellipse (0.2 and 1);
\node at (1,0){$Y$};

\draw (2,0) ellipse (0.2 and 1);
\node at (2,0){$X$};

\draw (3,0) ellipse (0.2 and 1);
\node at (3,0){$R$};

\draw[path2=20,-] (3,0.5) -- (2,0.5);\draw[path2=20] (2,0.5) node[point]{} --
(1,0.5) -- (0,0.5);
\draw[path2=20,-] (2,-0.5) -- (1,-0.5);\draw[path2=20] (1,-0.5) node[point]{} --
(0,-0.5);

\end{tikzpicture}
\caption{Illustration for Lemma \ref{transfer-up}}
\label{fig_lemma_transfer-up}
\end{center}
\end{figure}
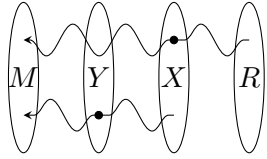

The following lemma is not used directly in further proofs, but serves
as an illustration of the techniques that will be used later.

\begin{lemma} \label{lemma_blocking-diff}
If $A_1$ blocks $X \to M$ and $X$ blocks $A_2 \to M$, then
$A_1 \setminus A_2$ blocks $X \to M$.
\end{lemma}
\begin{proof}
  The situation is illustrated in
  Figure~\ref{fig_lemma_blocking-diff}.  Let $A = A_1 \setminus
  A_2$. Assume to the contrary that~$A$ does not block $X \to M$, so
  there is an $A$-free path~$P$ from~$X$ to~$M$.  Since~$A_1$ blocks
  $X \to M$, these must be a vertex on this path which is
  in~$A_1$. Let~$w$ be the last such vertex on~$P$ and note that
  $w\in A_1 \cap A_2$ since~$P$ is $(A_1 \setminus A_2)$-free.  But,
  as~$X$ blocks $A_2 \to M$, there must be a vertex $u \in X$ on the
  part of~$P$ strictly after~$w$.  The suffix of~$P$
  starting from~$u$ is then a path connecting~$X$ with~$M$ and
  avoiding~$A_1$ (by the choice of~$w$), a contradiction to our assumption that $A_1$ blocks
  $X\to M$.
\end{proof}

\begin{figure}
\begin{center}
\begin{tikzpicture}

\draw (0,0) ellipse (0.2 and 1);
\node at (0,0){$M$};

\filldraw[rounded corners,fill=gray!10!white] (0.7,0) -- (0.7,1) -- (2,2) --
(3,2) -- (3,1.5) -- (2,1.5) -- (1.3,1) -- (1.3,-1) -- (0.7,-1) -- (0.7,0);
\node at (1,0){$A_1$};

\draw (2,0) ellipse (0.2 and 1);
\node at (2,0){$X$};

\filldraw[rounded corners,fill=gray!30!white,opacity=80] (2.7,0) -- (2.7,1.3) --
(2,1.3) -- (2,1.8) -- (3.3,1.8) --  (3.3,-1) -- (2.7,-1) -- (2.7,0) ;
\node at (3,0){$A_2$};

\draw[path2=20,-] (3,0.5) -- (2,0.5);\draw[path2=20] (2,0.5) node[point]{} --
(1,0.5) -- (0,0.5);
\draw[path2=20,-] (2,-0.5) -- (1,-0.5);\draw[path2=20] (1,-0.5) node[point]{} --
(0,-0.5);

\end{tikzpicture}
\caption{Illustration for Lemma \ref{lemma_blocking-diff}}
\label{fig_lemma_blocking-diff}
\end{center}
\end{figure}
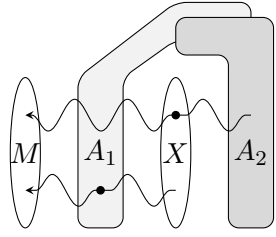

We formulate the following simple observation as a lemma.

\begin{lemma}\label{lemma_inclusion_minimal_blockers}
  Let $X$ be an inclusion-minimal set that blocks $R\to M$.
  Then for each $v\in X$ there is a path $P$ from $R$ to $M$
  such that $P \cap X = \{v\}$.
\end{lemma}

The following is our main technical lemma on blocking.

\begin{lemma}\label{lemma_min_extract}
Let $A$ and $B$ block $R\to M$. Then
\begin{enumerate}[(1)]
    \item\label{case_A_blocks_B} $A$ blocks $B\to M$ or
    \item\label{case_ex_B^*} there exists a set $B^* \subseteq A \cup B$ with
      $|B^*|<|B|$ which blocks $R\to B,M$, or
    \item\label{case_ex_A_new} there exists a set $A^*\subseteq A \cup B$ with
      $|A^*| \le |A|$ which blocks $A,B,R\to M$.
\end{enumerate}
\end{lemma}

\begin{proof}

We partition the set $B$ into elements $B_{free}$ from which $M$
is reachable via paths which avoid $A$ and the rest,
called $B_{rest}$, so $A$ blocks $B_{rest}\to M$. Moreover, we let
$A'$ be any \emph{inclusion-minimal} subset of $A$ such that
the set $A' \cup B_{rest}$ blocks $R\to B_{free}$. Observe that if 
$A' = \emptyset$, then either $B_{free} = \emptyset$, in which case $A$
already blocks $B\to M$ and we are done by \cref{case_A_blocks_B}, 
or $B_{rest}$ blocks $R\to B_{free} \neq \emptyset$ and thus $R\to M$, as
$B$ blocks $R\to M$, 
in which case $B^* = B_{rest}$ is the set we require in
\cref{case_ex_B^*}. We will now
consider the case when $A' \neq \emptyset$. This situation is
depicted in Figure \ref{fig_min_extract}. First observe a simple fact.

\begin{claim}\label{claim_B_rest_and_A'}
    For every $a\in A'$ there is a $B\cup (A'\setminus \{a\})$-free
    path from $R$ to $a$.
\end{claim}
\begin{proof}
As $B_{free} \cup A'$ blocks $B_{free}\to M$, by
Lemma~\ref{lemma_inclusion_minimal_blockers}, 
there is a path $P$ from $B_{free}$ to $M$ with $P\cap (B_{rest}\cup A') =
\{a\}$.
The suffix of $P$ from the last occurrence of $a$ is a path with the desired
properties: it never visits $B_{rest}\cup A'$ and thus also never visits
$B_{free}$, as  $B_{rest}\cup A'$ blocks $B_{free}\to M$.
\renewcommand{\qedsymbol}{$\dashv$}
\end{proof}

We will consider two cases. 

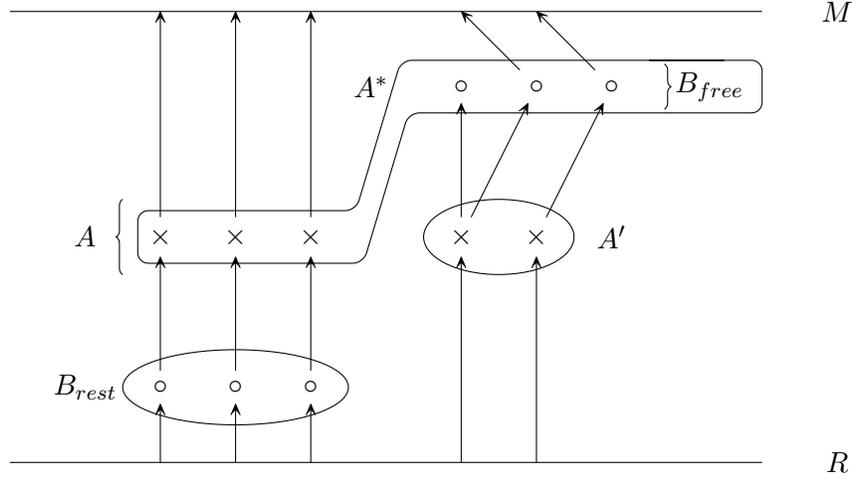
\begin{figure}
\begin{center}
\begin{tikzpicture}
\draw (-5, 3) -- (5, 3);
\node (c) at (6, 3) {$M$};
\draw (-5, -3) -- (5, -3);
\node (c) at (6, -3) {$R$};

\node (b0) at (-3, -2) {$\circ$};
\node (b1) at (-2, -2) {$\circ$};
\node (b2) at (-1, -2) {$\circ$};
\draw (b1) ellipse (1.5 and 0.5);
\node (bfree) at (-4, -2) {$B_{rest}$};


\draw[brace] (3.7,2.3) -- (3.7,1.7);
\node (bfree) at (4.3, 2) {$B_{free}$};
\node (b3) at (1, 2) {$\circ$};
\node (b4) at (2, 2) {$\circ$};
\node (b5) at (3, 2) {$\circ$};

\node (a0) at (-3, 0) {$\times$};
\node (a1) at (-2, 0) {$\times$};
\node (a2) at (-1, 0) {$\times$};
\node (a3) at (1, 0) {$\times$};
\node (a4) at (2, 0) {$\times$};

\draw (1.5, 0) ellipse (1 and 0.5);
\node (ab) at (3, 0) {$A'$};
\draw[brace] (-3.5,-0.5) -- (-3.5,0.5);
\node (A) at (-4, 0) {$A$};
\node (A_star) at (-0.2, 2) {$A^*$};

\draw[-slim] (a0) -- (-3, 3);
\draw[-slim] (a1) -- (-2, 3);
\draw[-slim] (a2) -- (-1, 3);

\draw[-slim] (-3, -3) -- (b0);
\draw[-slim] (-2, -3) -- (b1);
\draw[-slim] (-1, -3) -- (b2);

\draw[-slim] (b0) -- (a0);
\draw[-slim] (b1) -- (a1);
\draw[-slim] (b2) -- (a2);

\draw[-slim] (a3) -- (b4);
\draw[-slim] (b4) -- (1, 3);

\draw[-slim] (1, -3) -- (a3);
\draw[-slim] (2, -3) -- (a4);

\draw[-slim] (a3) -- (b3);
\draw[-slim] (a4) -- (b5);
\draw[-slim] (b5) -- (2, 3);

\draw[rounded corners] (3.5,2.35) -- (5,2.35) -- (5,1.65) -- (0.3,1.65) --
(-0.3,-0.35)
   -- (-3.3,-0.35) -- (-3.3,0.35) -- (-0.4,0.35) -- (0.2,2.35) -- (4.5,2.35);

\end{tikzpicture}
\end{center}
\caption{Situation in the proof of Lemma \ref{lemma_min_extract}.}
\label{fig_min_extract}
\end{figure}

\emph{Case (i):} $|A'| < |B_{free}|$.\\
Define $B^* = A' \cup B_{rest}$ -- it is smaller than $B$ and blocks $B\to M$,
which is \cref{case_ex_B^*}.

\emph{Case (ii):} $|A'| \geq |B_{free}|$.\\
Define $A^* = B_{free} \cup (A \setminus A')$. We claim that $A^*$ blocks 
$R\to M$. Assume to the contrary that there is a path $P$ from $R$ to $M$ which
avoids $A^*$. Since it avoids $B_{free}$ and $B$ blocks $R\to M$,
this path must go through $B_{rest}$. But, since $A$ blocks $B_{rest}\to M$,
it must visit  $A$ after each visit of $B_{rest}$. Let $a\in A$ be the last such
vertex. Since the path omits $A^*$, we have $a \in A'$.
By Claim~\ref{claim_B_rest_and_A'}, there is a $B\cup (A'\setminus \{a\})$-free
path $P'$ from $R$ to $a$.
Concatenating the $P'$ and the suffix of $P$ from $a$ we get
a $B$-free path from $R$ to $M$, which
contradicts the fact that $B$ blocks $R\to M$. 
Thus $A^*$ blocks $M$ from $R$.

Now we show that $A^*$ blocks $A,B\to M$.
First, $A^*$ blocks $A'\to M$, otherwise there is a $B_{free}$-free path
$P_0$ from $A'$ to $M$. Let $a'$ be the last vertex from $A'$ on $P_0$.
According to Claim~\ref{claim_B_rest_and_A'}, there is a $B$-free path $P_1$
from $R$ to $a'$. The concatenation of $P_1$ and the suffix of $P_0$ from $a'$
is $B$-free path from $M$ to $R$, which contradicts the assumption that
$B$ blocks $R\to M$. It follows that $A^*$ blocks $R\to A$. 

To see that $A^*$ also blocks $R\to B$,
note that $A$ blocks $B_{rest}\to M$ and $B_{free}\subseteq A^*$.
Finally, $|A^*| \le |A|$ since
$B_{free}$ is disjoint with $A$ by its definition.
 \end{proof}

\subsubsection*{A preorder on blocking sets}

The blocking relation induces a partial preorder on sets blocking $R \to M$.

\begin{definition} \label{def-blocking-order}
Let $A$ and $B$ block $R \to M$ in $G$.
We write $A \prec_M^R B$ if either $|A| < |B|$,
or $|A| = |B|$ and $A$ blocks $B \to M$.
\end{definition}

Intuitively, the second condition for $A \prec_M^R B$ means that~$A$
blocks from~$R$ as few vertices in addition to~$M$ as possible.
From Lemma~\ref{lemma_min_extract} we immediately obtain the following

\begin{corollary}\label{cor_excluding_A*}
  If $A$ is $\prec_M^R$-minimal and $B$ blocks $R\to M$, then
\begin{enumerate}[(1)]
\item\label{case_A_blocks_B} $A$ blocks $B \to M$ or
\item\label{case_ex_B^*} there exists a set $B^*$ with $|B^*|<|B|$ which
  blocks $R\to B$ and $R\to M$.
\end{enumerate}
\end{corollary}
\begin{proof}
  Assume that the case $(3)$ from Lemma~\ref{lemma_min_extract} holds.
  Let $A^*$ be a set with $|A^*|\le |A|$ that blocks $A\to M$,
  $B\to M$, and $R\to M$.
  In particular, $A^*$ blocks $A\to M$, so $A$ is not $\prec_M^R$-minimal.
\end{proof}

%

From Lemma~\ref{transfer-up} we obtain that $\prec_M^R$ is transitive,
so it is a preorder.
Moreover, Corollary~\ref{cor_excluding_A*} allows us to show the following
lemma.

\begin{lemma}\label{lemma_maximal-element}
There is a unique minimal element with respect to $\prec_M^R$.
\end{lemma}

\begin{proof}
  Assume that there exist two distinct $\prec_M^R$-minimal sets
  $A$ and $B$ that block $R\to M$, then neither $A \prec_M^R B$
  nor $B \prec_M^R A$. That means, $|A|=|B|$.
  Consider the cases given by Lemma~\ref{lemma_min_extract}.
  In \cref{case_A_blocks_B}, $A$ blocks $B \to M$, so $A\prec^R_M B$
  and $B$ is not minimal. In \cref{case_ex_B^*}, $B^*\prec^R_M B$,
  so $B$ is not minimal as well. In \cref{case_ex_A_new},
  $A^*\prec^R_M A$, so $A$ is not minimal. 
\end{proof}

We will denote the minimal element with respect to $\prec_M^R$, the
\emph{minimal blocker} of $R \to M$, by $\mb(R, M)$.

During the game, it is important to us how minimal blocking sets
behave when~$R$ becomes smaller or~$M$ becomes bigger, especially in
comparison to possible previous blocking sets. The next lemma allows
to compare a minimal set to a possibly non-minimal one.

\begin{lemma}\label{lemma_making_R_smaller}
  Let~$A$ be $\mb(R,M)$, let $R'\subseteq \Reach_{G-A}(R)$ and, for
  the new~$R'$, let~$A'$ be $\mb(R',M)$. Then~$A'$ blocks $R'\to A$.
\end{lemma}
\begin{proof}
Let $B = \Reach_{G}(R') \cap A$. It suffices to prove that $A'$ blocks
$R' \to B$. Apply Lemma \ref{lemma_min_extract} with
Corollary~\ref{cor_excluding_A*} to $B$ (as $A$), $A'$ (as $B$),
$R'$ (as $R$) and $M$ (as $M$). Consider \cref{case_A_blocks_B}.
Assume that there is a path $P$ from $R'$ to a vertex $b\in B$ that
avoids $A'$. By Lemma~\ref{lemma_inclusion_minimal_blockers} there is
an $A\setminus \{b\}$-free path $P'$ from $b$ to $M$. Concatenating $P$ with
the suffix $P'$ from the last occurrence of $b$ in that path we get
a path from $R'$ to $M$. As $A'$ blocks $R'\to M$, this path goes through $A'$.
As $P$ does not, there is some  $a'\in P'\cap A'$. As $B$ blocks
$A'\to M$, $P'$ visits $B$ after $a'$. As $P'\cap B=\{b\}$,
$P'$ visits $b$ after $a'$, but by definition of $P'$,
it contains~$b$ only as the first vertex, which is not $a'\in A'$,
because $a'\in P$ and $P$ is $A'$-free.

In Case~\ref{case_ex_A_new}, some set $B^*$ with $|B^*|<|B|$
blocks $A',B,R'\to M$, but $A$ is minimal, so if $B^*\neq B$,
we can replace $B$ by $B^*$ in $A$ to get a blocker $R'\to M$
(because $B^*$ blocks $B\to M$) with $B^*\cup(A\setminus B)
\prec_M^{R'} A$. This is impossible, since~$A$ is minimal, so $B^*=B$.
Thus $B$ blocks $A'\to M$ and we have Case~\ref{case_A_blocks_B}. 
\end{proof}

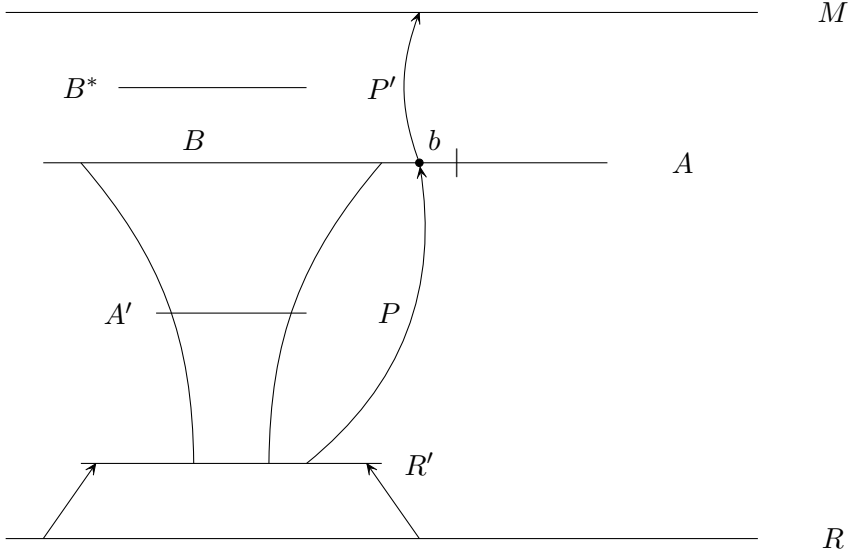
\begin{figure}
 \begin{tikzpicture}
 
\draw (-5,7) -- (5,7);
\node at (6,7){$M$};

\draw (-4.5,5) -- (3,5);
\node at (4,5){$A$};

\draw (-5,0) -- (5,0);
\node at (6,0){$R$};

\draw (-4,1) -- (0,1);
\node at (0.5,1){$R'$};

\draw [-slim](-4.5,0) -- (-3.8,1);
\draw [-slim](0.5,0) -- (-0.2,1);

\node (B_marker) at (1,5){$|$};
\node at (-2.5,5.3){$B$};


\draw  (-2.5,1) to [-slim,bend right=20] (-4,5);
\draw  (-1.5,1) to [-slim,bend left=20] (0,5);

\draw (-3,3) -- (-1,3);
\node at (-3.5,3){$A'$};

\node[point] (b) at (0.5,5){};
\node at (0.7,5.3){$b$};
\draw [-slim,bend right] (-1,1) to (b);
\node at (0.1,3){$P$};

\draw [-slim,bend left=20] (b) to (0.5,7);
\node at (0,6){$P'$};

\draw (-3.5,6) -- (-1,6);
\node at (-4,6){$B^*$};

\end{tikzpicture}
\caption{Illustration to Lemma~\ref{lemma_making_R_smaller}. }
\end{figure}

A similar result is obtained for the case when~$M$ grows.

\begin{lemma}\label{lemma_making_M_larger}
  Let $A$ be $\mb(R,M)$, let $M' \supseteq M$ and, for the new~$M'$,
  let~$A'$ be $\mb(R,M')$.  Then~$A'$ blocks $R\to A$ and~$A$ blocks
  $A'\to M$ .
\end{lemma}
\begin{proof}
Consider $B = \{a\in A'  |  M\cap \Reach_{G}(a)\neq \emptyset\}$ 
(\ie $B$ is the part of $A'$ from which $M$ is reachable) and apply
Lemma \ref{lemma_min_extract} with Corollary~\ref{cor_excluding_A*}
to $A$ and $B$. Case~\ref{case_ex_B^*} is impossible (replace $B$ by $B^*$
in $A'$, then $B^*\cup (A'\setminus B)$ blocks $R\to M'$ and
$B^*\cup (A'\setminus B)\prec^R_{M'} A'$, but $A'$ is $\prec^R_{M'}$-minimal),
so we have
Case~\ref{case_A_blocks_B}, \ie $A$ blocks $B\to M$.
Then $A$ blocks $A'\to M$, which shows the second statement.

Assume that there exists a path $P'$ from $R$ to $A$ that avoids $B$
(and thus $A'$). Let $a$
be the last vertex on this path.  By
Lemma~\ref{lemma_inclusion_minimal_blockers}, there is a path $P$ from
$R$ to $M$ whose intersection with $A$ is $\{a\}$. Consider the suffix
$S$ of $P$ from the last appearance of $a$. It does not visit $B$, as
$A$ blocks $A'\to M$ and thus also $B\to M$, so after each visit of $B$, $S$ would visit $A$, but
$P$ intersects $A$ only in $a$ and $S$ does not visit $a$ by
definition.  Concatenating $P'$ with $S$ we get a path from $R$ to $M$
that avoids $A'$, which is impossible, as $A'$ blocks $R\to M$. 
\end{proof}

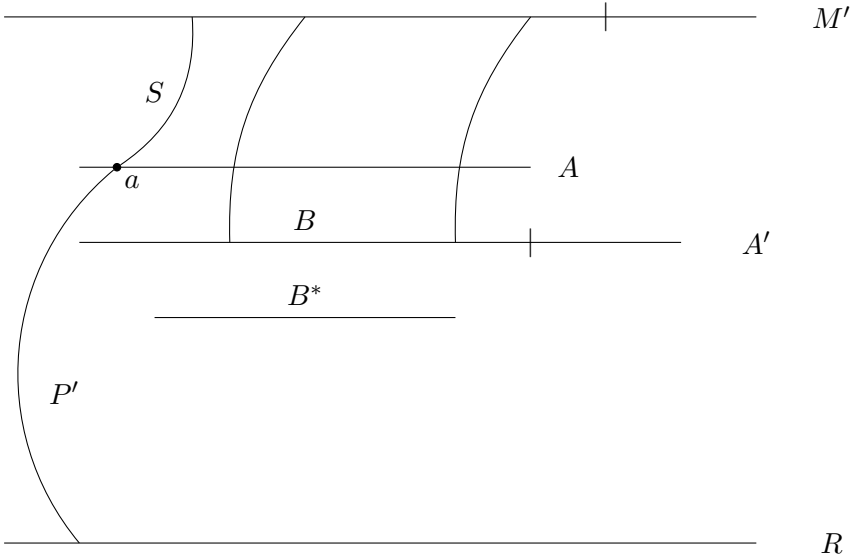
\begin{figure}
 \begin{tikzpicture}
 
\draw (-5,0) -- (5,0);
\node at (6,0) {$R$};

\draw (-4,4) -- (4,4);
\node at (5,4) {$A'$};

\draw (-5,7) -- (5,7);
\node at (6,7) {$M'$};

\node at (2,4){$|$};
\node at (-1,4.3) {$B$};

\draw (-3,3) -- (1,3);
\node at (-1,3.3) {$B^*$};

\node at (3,7) {$|$};
\draw (-2,4) to [bend left=20,-slim] (-1,7);
\draw (1,4) to [bend left=20,-slim] (2,7);

\draw (-4,5) -- (2,5);
\node at (2.5,5){$A$};

\node[point] (a) at (-3.5,5){};
\node at (-3.3,4.8){$a$};

\draw (-4,0) to [bend left=45,-slim] (a);
\node at (-4.2,2){$P'$};

\draw (a) to [bend right,-slim] (-2.5,7);
\node at (-3,6){$S$};

\end{tikzpicture}
\caption{Illustration to Lemma~\ref{lemma_making_M_larger}.}
\label{fig_cor_making_M_larger}
\end{figure}


    \subsection{Minimally blocking strategies}\label{sec_min_strat}

In this section we concentrate on a specific kind of strategies for the cops
in the shy weakly monotone game, namely ones that move chasers in the same
way, but whose guarding moves are placing the cops on the minimal
blocking set.

Let $\sigma$ be a strategy for the cops in the weakly monotone shy
game on~$G$.  We define the \emph{minimally blocking strategy}
$\sigma_\mb$, derived from $\sigma$, for possibly more cops than
$\sigma$, by induction on the length of play prefixes. This
construction also provides a function that maps each history~$\pi$
consistent with~$\sigma$ to a history
$\pi_\mb$ of the same length as~$\pi$ that is consistent with
$\sigma_\mb$ such that the following invariants
hold. Let $\pi = (C_0,R_0),(C_0,C_1,R_0), \ldots,
P$ where $P = (C_i,R_i)$ or $P = (C_i,C_{i+1},R_i)$, and let
$\sigma_\mb = (C_0^{\mb}, R_0^{\mb}), (C_0^{\mb}, C_1^{\mb},
R_0^{\mb}), \ldots, P_\mb$ where $P_\mb = (C_i^\mb,R_i^\mb)$ or $P =
(C_i^\mb,C_{i+1}^\mb,R_i^\mb)$.
\begin{enumerate}[(i)]
\item $R_i = R_i^\mb$ and $M(\pi) = M(\pi_\mb)$,
\item after a cop move, \ie if $P = (C_i,C_{i+1},R_i)$, we have $C_{i+1} \cap R_i = C^\mb_{i+1}
\cap R^\mb_i$ (the chasers are placed in the same way),
\item after a cop move, $\mb(M(\pi_\mb),R_i) \subseteq C_{i+1}^\mb $
  (the cops occupy the minimal blocker).
\end{enumerate}

Let $\pi[i]$ be the prefix of $\pi$ up to position $(C_i,R_i)$.
In the first move of the cops, if
$\sigma(\pi([0])) = \sigma((\0,R_0)) = (\emptyset,C_1,R_0)$,
with chasers $C^c_1 = C_1 \cap R_0$, then we set
\[ \sigma_\mb(\pi_\mb[0]) =
      (\emptyset, C^c_1 \cup \mb(R_0, C^c_1),R_0), \]
\ie we put the chasers and the minimal blocker. Obviously, the
invariants hold.

We turn to the inductive step. If the robber is the next to move, then
the last position in~$\pi$ has the form $(C_i,C_{i+1},R_i)$ and the
next move is to $(C_{i+1},R_{i+1})$ where $R_{i+1}$ is a strongly
connected component of $(G - C_{i+1})$ with $R_{i+1} \cap R_i\neq
\0$. As the cops play according to a robber-monotone strategy, we even
have $R_{i+1} \subseteq R_i$.  By the inductive hypothesis, the last
position in $\pi_\mb$ has the form $(C^\mb_i,C^\mb_{i+1},R^\mb_i)$ and
$C^\mb_{i+1} \cap R^\mb_i = C_{i+1} \cap R_i$. Thus, \emph{in the shy
  game}, the robber has exactly the same choices for $R_{i+1}^\mb$
from this position as from the end of~$\pi$. Therefore we can extend
$\pi_\mb$ by $(C^\mb_{i+1}, R^\mb_{i+1})$ and the conditions (i)--(iii) are satisfied.

Consider now the case that the cops are to move at the end of~$\pi$,
\ie the last position in~$\pi$ has the form $(C_i,R_i)$.  Let
$\sigma(\pi[i]) = (C_i,C_{i+1},R_i)$, with chasers $C^c_{i+1} = C_{i+1} \cap
R_i$. We set
\[ \sigma_\mb(\pi[i]) = (C^\mb_i, C^c_{i+1} \cup \mb(R^\mb_i,
M(\pi[i])) \cup \mb(R^\mb_{i-1}, M(\pi[i-1])), R^\mb_i)\,.
\]
Intuitively, we place the same chasers as~$\sigma$, occupy the current
minimal blocker and additionally the previous minimal blocker. It is
straightforward to see that all conditions (i)--(iii) are satisfied.

The construction above defines the strategy $\sigma_\mb$ and the corresponding
plays, but we are, of course, interested in strategies that are still weakly
robber-monotone. Strategy $\sigma_\mb$ is even strongly robber-monotone.

\begin{lemma}
Let~$\sigma$ be a strategy for cops in $\wmshyDAGWcops{G}$.
Then the strategy $\sigma_\mb$ is robber-monotone.
\end{lemma}
\begin{proof}
Every blocking set
  $\mb(R_i,M(\pi[i]))$ blocks $R_i \to M(\pi[i])$, so vertices that
  have been occupied by the chasers are not available for the robber. By 
Lemma~\ref{lemma_making_R_smaller} and
  Lemma~\ref{lemma_making_M_larger} previous blocking sets are blocked
  by later blocking sets as $R_i$ becomes smaller and $M(\pi[i])$
  becomes bigger.
\end{proof}

Let us calculate the number of cops used by $\sigma_\mb$.

\begin{lemma}
Let $\sigma$ be a winning strategy for $k$ cops in the weakly monotone
shy game on~$G$.
Then $\sigma_\mb$ is a winning strategy for $3k$ cops in the (strongly
monotone) shy game on~$G$.
\end{lemma}
\begin{proof}
The strategy $\sigma_\mb$ is monotone by the previous lemma, and by
property~(i) of the definition, the components available for the robber
correspond to those in plays consistent with~$\sigma$, thus~$\sigma_\mb$
is winning for the cops. To calculate the number of cops used by~$\sigma_\mb$,
recall that the set of cops placed in step~$i$ is
$C^\mb_{i+1} = C^c_{i+1} \cup \mb(R^\mb_i, M(\pi[i])) \cup \mb(R^\mb_{i-1}, M(\pi[i-1])), R^\mb_i)$,
where $C^c_{i+1}$ were the chasers placed by~$\sigma$, \ie the set $C_i \cap R_i$,
where $C_i$ are all cops placed by $\sigma$ in the corresponding position.
Since $\sigma$ was a weakly monotone strategy, the set $C_i$ blocks
$R_i \to M(\pi[i])$, and the previous $C_{i-1}$ blocked $R_{i-1} \to M(\pi[i-1])$.
Thus $|\mb(R_i,M_i)| \le |C_i| \leq k$ and $|\mb(R_{i-1},M_{i-1})| \le
|C_{i-1}|\leq k$, and, of course, $|C^c_{i+1}| \leq k$. Therefore $|C^\mb_{i+1}| \leq 3k$.
\end{proof}

\begin{corollary}\label{cor_compare_games}
$\wmshyDAGWcops{G} \leq \shyDAGWcops{G} \leq 3 \cdot \wmshyDAGWcops{G}$
\end{corollary}

To convince oneself that these inequalities are not trivial, and that
blocking minimally makes a difference, consider the following lemma.

\begin{lemma}\label{lemma_guarding-sequences}
There are graphs on which the cops have to make more than one guarding moves
successively in order to win in the weakly monotone shy game, respectively
in the strongly monotone game with the least possible number of cops.
\end{lemma}
\begin{proof} 
Let $n\ge 6$.
Consider the graph $G_n$ depicted in
Figure~\ref{fig_guarding_sequences} (recall that edges without arrows
denote edges in both directions).
Arrows that connect parts of the graph enclosed in a rectangle lead to or
from all vertices of the graph.
The graph consists of a vertex $c_0$ and $n$ parts $A_i$ that are isomorphic
to each other and connected only to $c_0$ and in the same way.
Every $A_i$ consists of a $2$-clique with vertices labeled in the picture
with $c_1$ and a $3$-clique with vertices labeled in the picture with
$c_2$ that are connected to each other and to $c_0$.
Further, $A_i$ contains $n$ parts $B_j$.
Each $B_j$ contains a $3$-clique $R$, a single vertex labeled with $g_0$
and a $2$-clique with vertices labeled with $g_1$.
The connections are shown in the figure.

If the cops are allowed to make multiple guarding moves in a row,
$6$ of them suffice to capture the robber strongly (and thus weakly)
monotonously. One cop is placed on $c_0$ (a chasing move), the robber
chooses a component $A_i$. Then the cops occupy vertices $c_1$ and $c_2$ in 
further chasing moves, the robber chooses a part $B_j$ in $A_i$.
Then the cop from $c_0$ goes to $g_0$, which is a guarding move,
and then the cops from $c_1$ go to the both vertices $g_1$. 
Note that if the robber remains in $R$, placing cops on $g_1$ is
again a guarding move. Finally, the cops from $c_2$ capture the robber in $R$.
Note that if there are $7$ cops, it is possible to place this additional
cop on a vertex in $R$ instead of making the second guarding move
and then win as before.

If the cops are not permitted to make two guarding moves in a row,
the robber has the following winning strategy in the weakly (and thus strongly)
monotone game against $6$ cops.
In the first move, the robber occupies $c_0$ and waits there until
it is occupied by a cop. In that moment, there is a cop free component $A_i$
(as there are $6$ cops and $6$ components $A_0,\dots,A_{n-1}$,
but one cop occupies $c_0$). The robber goes to that cop free component $A_i$
and waits on the $5$-clique that is build by vertices $c_1$ and $c_2$.
When the cops occupy this clique, there is a cop free part $B_j$ in $A_i$
and the robber runs there. Note that cops on the clique are chasers,
so the only free cop is that from $c_0$. If he is placed in $R$,
the robber stays, some other cop must move up, and the cops lose.
If he is placed on one of the vertices $g_1$ (which is a guarding move),
as no second guarding move is allowed, the only possible next move for cops is
to place the one from $g_1$ on $R$ -- and lose as before. Hence, we can assume
that the cop from $c_0$ is placed on $g_0$, a guarding move. The next move
must be chasing and the only possibility is to place a cop from $c_1$ in $R$.
Now the cop on $g_0$ cannot be taken away, as a path to $c_1$ would be cop free,
and the cops on $c_2$ are still bound as well. So is the cop in $R$
(his move was chasing). Thus there is only one free cop (on $c_1$).
He makes a guarding move, then a chasing move to $R$ and the cops lose.
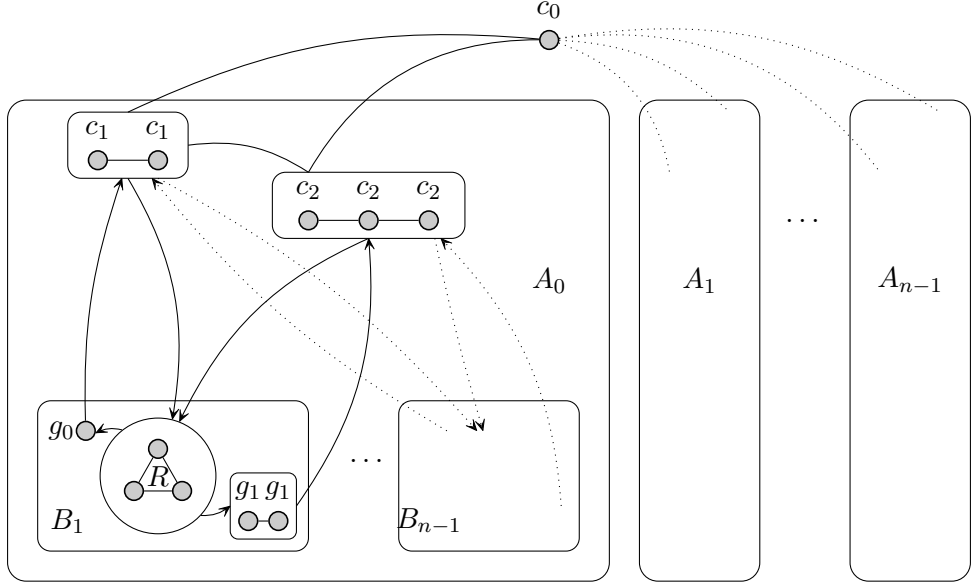
\begin{figure}
\begin{center}
\begin{tikzpicture}[scale=0.8]

\node[vertex] (c_0) at (1,5){}; 
\node at (1,5.5) {$c_0$};

\draw [rounded corners=7pt]  (-8,-4) rectangle (2,4); 
\node at (1,1){$A_0$};

\draw [rounded corners=7pt]  (2.5,-4) rectangle (4.5,4);
\node at (3.5,1){$A_1$};
\draw (c_0) edge [dotted,bend left=30] (3,2.8);
\draw (c_0) edge [dotted,bend left=20] (4,3.8);

\node at (5.25,2){\dots};

\draw [rounded corners=7pt]  (6,-4) rectangle (8,4);
\node at (7,1){$A_{n-1}$};
\draw (c_0) edge [dotted,bend left=30] (6.5,2.8);
\draw (c_0) edge [dotted,bend left=20] (7.5,3.8);


\node[vertex] (c_1_left) at (-6.5,3){};
\node at (-6.5,3.5){$c_1$};
\node[vertex] (c_1_right) at (-5.5,3){};
\node at (-5.5,3.5){$c_1$};
\draw [] (c_1_left) -- (c_1_right);
\draw [rounded corners=5pt] (-7,2.7) rectangle (-5,3.8);

\node [vertex] (c_2_left) at (-3,2){};
\node at (-3,2.5){$c_2$};
\node [vertex] (c_2_middle) at (-2,2){};
\node at (-2,2.5){$c_2$};
\node [vertex] (c_2_right) at (-1,2){};
\node at (-1,2.5){$c_2$};
\draw (c_2_left) -- (c_2_middle) -- (c_2_right);
\draw [rounded corners=5pt] (-3.6,1.7) rectangle (-0.4,2.8);

\draw (-5,3.25) edge[bend left=20] (-3,2.8);
\draw (c_0) edge[bend right=15] (-6,3.8);
\draw (c_0) edge[bend right] (-3,2.8);

\draw [rounded corners=5pt] (-7.5,-3.5) rectangle (-3,-1);
\node at (-7,-3){$B_1$};

\node [vertex] (g_0) at (-6.7,-1.5){};
\node at (-7.1,-1.5){$g_0$};

\node [vertex] (g_1_left) at  (-4,-3){};
\node at (-4,-2.5){$g_1$};
\node [vertex] (g_1_right) at  (-3.5,-3){};
\node at (-3.5,-2.5){$g_1$};
\draw (g_1_left) edge (g_1_right);
\draw [rounded corners=3pt] (-4.3,-3.3) rectangle (-3.2,-2.2);

\node [draw,circle,minimum size=4em] (circle) at (-5.5,-2.25){$R$};
\node [vertex] (R_up) at (-5.5,-1.8){};
\node [vertex] (R_left) at (-5.9,-2.5){};
\node [vertex] (R_right) at (-5.1,-2.5){};
\draw (R_up) edge (R_left) edge (R_right);
\draw (R_right) edge (R_left);

\draw (g_0) edge [-slim,bend left=10] (-6.1,2.7); 
\draw (-3.2,-2.75) edge [-slim,bend right=20] (-2,1.7); 
\draw (circle) edge [-slim,bend right=20] (-4.3,-2.75); 
\draw (circle) edge [-slim,bend right=20] (g_0);
\draw (-6,2.7) edge [-slim,bend left=20] (circle);
\draw (-2,1.7) edge [-slim,bend right=20] (circle);

\node at (-2,-2){\dots};
\node at (-1,-3){$B_{n-1}$};
\draw [rounded corners=5pt] (-1.5,-3.5) rectangle (1.5,-1);
\draw (-0.7,-1.5) edge [dotted,-slim,bend left=10] (-5.6,2.7); 
\draw (1.2,-2.75) edge [dotted,-slim,bend right=20] (-0.8,1.7); 
\draw (-0.9,1.7) edge [dotted,-slim,bend right=2] (-0.1,-1.5); 
\draw (-5.5,2.7) edge [dotted,-slim,bend left=10] (-0.2,-1.5); 

\end{tikzpicture}
\caption{The cops need more than one guarding move in a row.}
\label{fig_guarding_sequences}
\end{center}
\end{figure}
\end{proof}


    \subsection{Decomposition}\label{sec_dec}

Our next goal is to define a decomposition of graphs in the spirit of
\cite{JohnsonRobSeyTho01} for the strongly monotone shy game.
Let~$G$ be a graph. A \emph{shy-monotone tree decomposition} of~$G$
is a tuple $(T,C,R)$ where~$T$ is a directed tree with root~$r$ 
and edges oriented away from the root, and $C, R : V(T) \to 2^G$ are
functions with the properties listed below, which intuitively correspond
to the placements of the cops and the component of the robber.
For a node $t \in V(T)$ we write $ch(t)$ for the set of chasers
corresponding to~$t$, \ie $ch(t) = C(t) \cap R(t)$, and we denote by
$g(t)$ the guards, $g(t) = C(t) \setminus R(t)$.
Moreover, we write $m(t)$ for the union $\bigcup_{s \preceq t} ch(s)$,
\ie for the set of all chasers from the nodes above~$t$ in~$T$.

\begin{enumerate}[(1)]
\item \label{shy_dec_root}
  For the root $r$, $R(r) = V(G)$.

\item \label{shy_dec_sc} For every $(t,t')\in E(T)$,
  $R(t')$ is a strongly connected component of $R(t) \setminus ch(t)$.

\item\label{shy_dec_cover_robber} For every $t\in V(T)$
  if $t_1, \dots, t_n$ are all direct successors of~$t$, then
  \[ R(t) = ch(t) \cup \bigcup_{i=1}^n R(t_i). \]

\item \label{shy_dec_block} For every $(t,t')\in E(T)$ holds:
  \[ m(t) \cap \Reach_{G - (C(t) \cap C(t'))} R(t') = \emptyset, \]
  \ie there is no path from $R(t')$ to $m(t)$ avoiding $C(t) \cap C(t')$.
\end{enumerate}

Note that from Item~\ref{shy_dec_root} and Item~\ref{shy_dec_cover_robber} it
follows that every vertex of~$G$ is contained in the image of~$ch$,
$\bigcup ch(V(T)) = V(G)$. Indeed, for a node~$t$ without successors,
we obtain from Item~\ref{shy_dec_cover_robber} that $R(t) = ch(t)$, and
applying this item inductively proves that, for each node~$t$,
the component $R(t)$ is covered by $\bigcup_{t \preceq t'} ch(t')$.
Since, by Item~\ref{shy_dec_root}, in the root $R(r) = V(G)$, we
have $\bigcup ch(V(T)) = G$.

The \emph{width} of a shy-monotone tree decomposition $(T,C,R)$ is defined
as $\max \big\{ |C(t)|  \mid  t \in V(T) \big\}$.

\begin{proposition}\label{prop_eq_defs_shy_and_dec}
Let $G$ be a graph. The following statements are equivalent.
\begin{enumerate}[(1)]
\item\label{lemma_def_shy} 
  $k$ cops capture the robber in the shy-monotone game on~$G$.
\item\label{lemma_def_dec}
  There is a shy-monotone tree decomposition $(T,C,R)$ of~$G$ of width~$k$.
\end{enumerate}
\end{proposition}

\begin{proof} ~
\begin{description}
\item[$(1) \Rightarrow (2)$.]  Let $\sigma$ be a strategy for~$k$ cops
  on~$G$.  We construct $(T,C,R)$ inductively, starting with the
  root~$r$ with $R(r) = V(G)$ and we set $C(r)$ to the first placement
  of the cops chosen by~$\sigma$. We continue the construction by
  following a play consistent with~$\sigma$ in each component chosen
  by the robber, and setting~$C(t')$ to the vertices occupied by cops
  placed if the robber makes the respective move to
  $R(t')$. Items~\ref{shy_dec_sc} and \ref{shy_dec_cover_robber}
  follow from the general definition of the game, while
  Item~\ref{shy_dec_block} follows from the game being weakly-monotone.

\item[$(2) \Rightarrow (1)$.]  From the decomposition $(T,C,R)$ we
  construct a strategy $\sigma_T$. The first move of the cops is to
  $C(r)$, where $r$ is the root of $T$. For each move of the robber to
  $R'$, the cops respond with the move to $C(t')$, where $t'$ is the
  successor of $t$ with $R(t') = R'$. Items~\ref{shy_dec_sc} and
  \ref{shy_dec_cover_robber} guarantee that this strategy is well
  defined, while Item~\ref{shy_dec_block} guarantees that it is
  winning. Obviously, the number needed cops is the width of the decomposition.
\end{description}
\end{proof}

We continue with an analysis of the decompositions. Let $\sigma$ be a
strategy for the cops in the shy-monotone game on $G$ and $T$ the
corresponding decomposition tree.  For a non-empty set of vertices $A$
there is a unique \emph{split vertex} $\split{A}$ which is the latest
common predecessor of all vertices of $A$.  We also write
$\split{a,b}$ for $\split{\{a,b\}}$ and $\split{a,A}$ for
$\split{\{a\}\cup A}$.


The proof of the next lemma is easy and we omit it.

\begin{lemma}\label{lemma_exactly_one_cop}
 If weak \dagw of a graph~$G$ is~$k$, then there is a winning strategy~$\sigma$
for the cops that always prescribes to place exactly one cop in a move,
\ie if $(M,C,R)\to(M',C,C',R)$ is a move according to~$\sigma$, then
$|C'\setminus C| =1$.
\end{lemma}

It follows that we can turn any shy-monotone tree decomposition into one with
$|ch(t)| = 1$ for all $t\in V(T)$.

\begin{corollary}\label{cor_exactly_one_cop}
 For every graph $G$ with $\wmshyDAGWcops{G} = k$ there is a shy-monotone
tree decomposition $(T,C,R)$ of width~$k$ where 
for all $t\in V(T)$, there is a vertex $w\in V(G)$ with $|ch(t)| = \{w\}$.
\end{corollary}


We define an order on the vertices of a graph that corresponds the
order in which the robber is chased in some plays which are played
according to~$T$.  Let $G$ be a graph and let $T$ be its
shy-{}monotone tree decomposition with $|ch(t)| = \{t\}$, for each
$t\in V(T)$. Let $v$ and $w$ be two vertices of the graph. We say
that~$w$ is \emph{to the right} of $v$ (and $v$ is \emph{to the left}
of $w$) and write $v\totheright w$ 
   if
\begin{enumerate}[1)]
\item $w$ is on the path from the root of $T$ to $v$, or
\item there is a path from $w$ to $v$ in $G - m(\split{w,v})$.
\end{enumerate}
In other words, $w$ is to the right of $v$ in $G$ if, in a position in
which a chaser occupies $w$, there is a cop free path from $w$ to $v$. In the
decomposition, we have in that case either $v\in R(w)$ and $w\notin R(v)$, or
$\split{v,w} \notin \{v,w\}$ and there is a path from $w$ to $v$ that avoids
vertices above $\split{v,w}$, see Figure~\ref{fig_right_to_left} for the
explanation of our terminology.

Clearly, $\totheright$ is a partial order. We abuse the notation and denote any
linearisation of $\totheright$ also by~$\totheright$.

\begin{figure}
\begin{center}
\begin{tikzpicture}

\draw (-5,-5) -- (0,0) -- (5,-5);
\node[vertex] (split) at (0,-1){};
\draw (0,0) -- (split);
\node[vertex] (w) at (2,-3){};
\node[vertex] (v) at (-2,-3){};
\draw (w) -- (split) -- (v);
\draw (w) -- (3,-5) -- (1,-5) -- (w);
\draw (v) -- (-1,-5) -- (-3,-5) -- (v);

\node[vertex] (right) at (2.2,-4.4){};
\node[vertex] (left) at (-2.1,-4){};
\draw[slim-,decorate,decoration=snake] (right) -- (left);

\node (splitlabel) at (-2,-0.5){$\split{v,w}$};
\draw[ie] (splitlabel) -- (split);
\node at (1.6,-3){$v$};
\node at (-1.6,-3){$w$};

\draw[brace] (0.6,0) -- (0.6,-1);
\node at (2,-0.5){$m(\split{v,w})$};
 
\end{tikzpicture}
\caption{The order $\totheright$ ''to the right of''.}
\label{fig_right_to_left}
\end{center}
\end{figure}
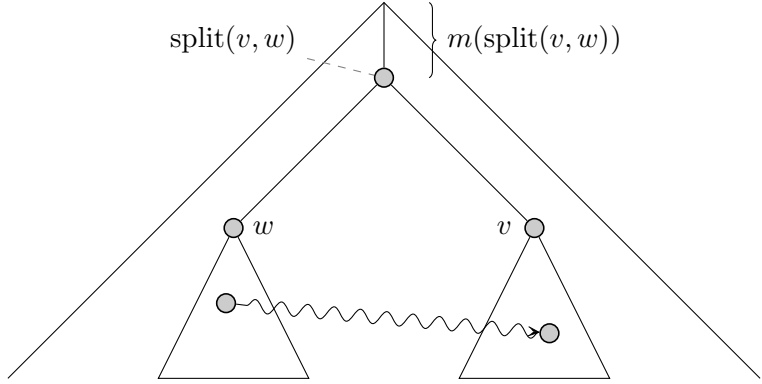

As a next step, we show a simple, but useful property of a
shy-monotone tree decomposition. Informally, the following lemma says
that there is no path from left to right in the decomposition tree
which avoids common predecessors of the first and the last vertices on
the path.

\begin{lemma}\label{lemma_no_left_to_right}
If $v\totheright w$ and $\split{v,w} \notin \{v,w\}$, then $m(\split{v,w})$
blocks $v\to w$.
\end{lemma}
\begin{proof}
Let $P$ be a path from $w$ to $v$ and let $P'$ be a path from $v$ to $w$. We
show that $P'\cap m(\split{v,w}) \neq \emptyset$. Let $u\in G$ be a vertex
in $P'$ such that $\split{u,w}$ has a minimal distance from the root of the
decomposition tree. Then $P' \subseteq R(\split{u,w})$. Indeed, if there is a
vertex $u'\in P' \setminus R(\split{u,w})$, then $\split{u',R(\split{u,w})}$ is
nearer to the root than $\split{u,w}$ (by Item~\ref{shy_dec_sc} of the
definition of a shy-{}monotone tree decomposition) and thus nearer than
$\split{u,w}$ (as $w\in R(\split{u,w})$) contradicting the choice of~$u$.

As $\split{v,w} \notin \{v,w\}$ and by the definition of $\split{\cdot}$, $u$
and $w$ are in different components of $R(\split{u,w}) - \split{u,w}$. As there
is a path from $w$ to $u$ (concatenate $P$ with with prefix of $P'$ up to $u$),
there is no path from $u$ to $w$ in $R(\split{u,w}) - \split{u,w}$, \ie
$\split{u,w} \in P'$, so $P'\cap m(u,w) \neq \emptyset$. By the choice of $u$,
we have $m(u,w) \subseteq m(v,w)$, so $P'\cap m(v,w) \neq \emptyset$.
\end{proof}

If the robber leaves his component, he moves from the right to the
left in the decomposition tree. By property~(2) he can return to his
left component only via $ch(t)$ for some. However $ch(t)$ are vertices
where there are or have been chasers, so a winning cop strategy does
not allow the robber to visit them. Thus he cannot return.

\begin{lemma}\label{lemma_no_left_to_right}
  For every winning strategy~$\sigma$ in the weakly monotone game
  in every play $\pi = (C_0,R_0), (C_0,C_1,R_0), \ldots$
  consistent with~$\sigma$, if the robber leaves a component~$R$ with
  a move $(C_i, C_{i+1}, R_i) \to (C_{i+1}, R_{i+1})$, then
  the cops on $C_i\cap C_{i+1}$ block
  $R_{i+1}\to R$. Thus the robber will never be able to enter~$R_i$ again.
\end{lemma}

It is not known whether determining the \dagw of a graph is solvable in non-deterministic
polynomial time. For weak \dagw, however, it is. The argument is that 
shy-monotone tree decompositions
have polynomial size in the size of the graph.

\begin{theorem}
 Given a graph $G$ and a natural number $k$, it is in NP to decide whether~$G$ has weak \dagw at most $k$.
\end{theorem}
\begin{proof}
  The algorithm guesses the decomposition tree, which has size
  $\calO(|G|^2)$ (because for each new chaser the guards have to 
be moved at most at most $|G|$ many times) and checks in polynomial time 
whether it is correct. 
\end{proof}

    \subsection{From shy to weakly monotone game, shy-similar strategies}\label{sec_wmon}

First we define some conditions on the players' strategies that
can be assumed without loss of generality.

\begin{definition}
A \emph{chasing (guarding) move} of cops is a move where only chasing
(guarding) cops are placed. (Note that both sorts of cops may be taken.)
A cop strategy is \emph{pure} if it consists only of guarding and chasing
moves (and has no mixed moves).
\end{definition}

\begin{lemma}\label{lemma_pure_str}
If~$k$ cops have a winning strategy, then~$k$ cops have a pure winning strategy.
\end{lemma}

\begin{proof}
Assume an arbitrary winning strategy $f$ for $k$ cops in the weakly monotone game
on a graph $G$. At the beginning, only chasing moves are possible.
Later on, instead of a mixed move $(C,v)\to(C,C',v)$ ,where $C_c =
('\setminus C)\cap \flap(v,C)$ are the new chasers and $C_g =
(C'\setminus C)\setminus\flap(v,C)$ are the new guards, make first the guarding 
part,
\ie place cops on $C_g$ of the move. If the robber changes his component, take the cops 
from $G$ away and continue translating the strategy as if the robber
went to the new component one move ago, \ie before the cops move
to~$C'$. This is possible, as the cops were guarding and thus did not
change the component and thus the resulting position
is still consistent with~$f$. Note that the number of times 
the robber changes his component is finite. So assume w.o.l.g. that
the robber remains in his component. Then the cops make the chasing part of their move,
\ie the cops are placed on $C$. It is easy to see that every robber move
leads to position that is consistent with $f$. Further, no strong
non-monotonicity occurs. Thus the new strategy is winning for the cops.
\end{proof}

\begin{lemma}[cf.~\cite{PuchalaRab10}, Lemma $8$]\label{lemma_cool_robber}
In a weakly monotone \dagw game, if the robber has a winning
strategy~$\sigma$ against~$k$ cops, then he also 
has a strategy that never
prescribes to change his component if no cop was placed on a vertex
reachable for the robber in the previous move.
\end{lemma}

\begin{proof}
Assume that $\sigma$ prescribes to move to a component $C$ although no cop was
placed in the reachability region of the robber. Change the strategy such that
the robber never moves in such positions. Obviously, some cop eventually must
be placed in the region, otherwise the robber wins. After this, the robber
can still move to the same component of the current position as from $C$. 
\end{proof}

\begin{proposition}\label{prop_strong_shy_weak}
If~$k$ cops have a winning strategy in the strongly monotone shy robber game
on~$G$, then $2k$ cops have a winning strategy in the weakly monotone game
on $G$.
\end{proposition}

\begin{proof} 
We say ``shy game'' for the strongly monotone shy robber game and ``weak game''
for the weakly monotone game.  We translate the moves of the robber from
the weak game to the shy game and the moves of the cops vice versa.
Let $\sigma$ be a pure winning strategy for the cops in the shy game.
We describe the new \emph{shy-similar} strategy $\shysim(\sigma)$ for the
weakly monotone game.

Consider a robber move $(M',C,C',R) \mapsto (M',C',R')$ in the weak game.
If $R' \subseteq R$, then we translate the move as
$(M',C,C',R) \mapsto (M',C',R')$ and take the next move according to $\sigma$, so
that nothing changes with respect to $\sigma$. Otherwise, \ie if $R'$ is not a subset or $R$,
we consider the latest move $(M_i,C_i,R_i)\mapsto(M_{i+1},C_i,C_{i+1},R_i)$ of
the cops, such that $R_i \supseteq R'$. (Since it is the latest such move,
we know that $R_{i+1} \not\supseteq R'$
.)
As $\sigma$ is strongly monotone, $M_i$ is blocked from $R'$ by $C' \setminus R$.
Furthermore, $M' \setminus C_{i+1}$, \ie the set of vertices where chasers have been
placed after position $(M_{i+1},C_i,C_{i+1},R_i)$, is not reachable from $R'$
either. Let us place all the guards for the position where $R'$ appears
in the continuation of the play from $(M_{i+1},C_i,C_{i+1},R_i)$ towards $R'$.
Due to Lemma~\ref{lemma_cool_robber} we can assume that the robber remains
on $R'$ during this time. After this move, we remove the other
guards (here weak non-{}monotonicity can occur)
and place the chasers as in the position for $R'$.
This is the only  place where (weak) non-monotonicity occurs.
We have the same position that would occur if the robber would have moved
to $R'$ in the shy-monotone game, and we continue to play $\sigma$ from there.
\end{proof}

\begin{definition}
A winning strategy $\sigma'$ is \emph{shy-similar} if there is a winning strategy $\sigma$
for the cops in the strongly monotone shy robber game such that $\sigma' = \shysim(\sigma)$
where $\shysim(\sigma)$ is the strategy that is constructed from $\sigma$ as shown in
Proposition~\ref{prop_strong_shy_weak}.
\end{definition}



\begin{corollary}
If~$k$ cops win the weakly monotone cops and robber game on~$G$,
then $6k$ cops win the weakly monotone game on~$G$.
\end{corollary}

\begin{proof}
  If $\wmDAGWcops{G}\le k$, then $\wmshyDAGWcops{G}\le k$ because the
  cops can use the same winning strategy. By
  Corollary~\ref{cor_compare_games}, $\shyDAGWcops{G}\le 3k$ and
by Proposition~\ref{prop_strong_shy_weak}, $\wmDAGWcops{G}
\le 6k$.
 \end{proof}	

%
%


    \subsection{Strongly monotone strategies: two attempts}\label{sec_fail_tries}

In this section, we use the decomposition defined above to construct
a strongly monotone winning strategy for a bounded number of cops. 
Our construction is a combination of
two approaches: leaving tied cops and freezing the context. Leaving cops
is a transformation of a strategy~$\sigma$ by not removing tied cops,
\ie those wo must be removed according to~$\sigma$, but
whose removal would immediately lead to strong non-{}monotonicity.
Freezing the context changes a given strategy by marking the current robber
component~$R$ and playing further only in~$R$, \ie omitting any
changes outside the component, until the robber leaves~$R$ or is
captured: the cops outside of~$R$ are ``frozen''. In particular, no
cops are placed outside the robber component. Obviously, both
transformations produce strongly monotone strategies, but may
use more cops than~$\sigma$. In the following we define both
approaches formally and show that, first, both taken independently lead to an
unbounded number of additional cops they introduce, but, second, they
can be combined into one transformation that uses only a quadratic
number of additional cops.

\subsubsection{Leaving tied cops is not enough} \label{subsec:lstay}

To make precise which cops are tied, we define the \emph{front} of
a subset~$X$ of vertices of a graph~$G$ with respect to~$R$. Let
$X,R\subseteq G$, $X\cap R = \0$. Then the front $\front_G(R,X)$ is the inclusion minimal
subset of~$X$ that blocks $R \to X$ in~$G$. If~$R = \{v\}$, we also
write~$\front_G(v,C)$. Let us prove that this set is unique.
Indeed, assume that two distinct minimal subsets $X_0 \subseteq X$ and $X_1
\subseteq X$ block $R \to X$. Then w.l.o.g. there is a vertex $v\in
X_0\setminus X_1$. As~$X_0$ is minimal, there is a
$X_0\setminus\{v\}$-free path from~$R$ to~$v$. As~$X_1$ blocks $R\to
X$, this path goes through a vertex $w\in X_1\setminus X_0$. However the
prefix of the path from~$R$ to~$w$ is $X_0$-free, which contradicts
that~$X_0$ blocks $R\to X$.

The \emph{leaving-cops} strategy $\sigma_\lc$\label{def_sigma_lc} is
as~$\sigma$, but it leaves the cops from $\front_G(v,C)$ on their
vertices. Here~$v$ is the robber vertex and~$C$ is the
placement of the cops. More formally, we define $\sigma_\lc$ as a memory strategy.
The memory stores the cop placement we would have
playing according to~$\sigma$. So a memory state is a set~$P\subseteq V$. 
Initially,~$P=\0$. When the robber moves,~$P$ does not change.
In a position~$(C,v)$ with a memory state~$P$, the new strategy prescribes
to move as if the position was~$(P,v)$, but removing only those cops that
are not reachable from the robber vertex. In other words, 
\[\sigma_\lc(C,v) = \front_G(v,C\cup \sigma(P,v))\,.\] 
Obviously, if~$\sigma$
is a strongly monotone winning strategy for~$k$ cops, then~$\sigma_\lc$
is strongly monotone winning strategy for~$k$ cops. If~$\sigma$ is a weakly 
monotone winning strategy, then so is~$\sigma_\lc$, but~$\sigma_\lc$ may use more cops.

It is not a priori clear whether there is a class of graphs and a
strategy~$\sigma$ such that~$\sigma$ uses a bounded and~$\sigma_\lc$ an unbounded number
of cops on graphs from that class. However, we show in this subsection
that~$\sigma_\lc$ can be arbitrarily bad compared to~$\sigma$.
The idea is to iterate the argument from~\cite{KreutzerOrd08},
with the (rough) correspondence between the graph $G_2$ in 
Figure~\ref{fig_against_leaving_cops} and their graph $D_p$ 
(Figure~1 in~\cite{KreutzerOrd08}) as follows. The component
$C_1$ in $G_2$ corresponds to $C_0$ in $D_p$, 
the component $R_1$ corresponds to $C_2$ in $D_p$,
$A_1$ corresponds to $C^1_1$, and finally
$C_2$ in $G_2$ to $C^2_1$ in $D_p$. Disregarding the sizes,
the only $D_p$-edges missing in $G_2$ are between $C^1_1$ and $C_0$,
which corresponds to connecting $A_1$ and $C_1$ in $G_2$.
While adding an edge from $C_1$ to $A_1$ is possible in $G_2$,
it is essential that no $A_1 \to C_1$ edge is present. But these edges,
corresponding to $C^1_1 \to C_0$ edges in $D_p$, are not important in $D_p$.

\begin{figure}
\begin{center}
\begin{tikzpicture}[scale=0.7]
\node[vertex,label=above left:$C_1$] (C_1)   at (0,10) {};
\node[vertex,label=above left:$C_2$] (C_2)   at (2,10) {};
\draw [-slim] (C_2) to (C_1);
\node (Cdots)                              at (4,10) {\dots};
\node[vertex,label=above left:$C_{n-1}$] (C_n-1) at (6,10) {};
\node[vertex,label=above left:$C_n$] (C_n)   at (8,10) {};
\draw [-slim] (C_n) to (C_n-1);

\draw [-slim] (C_n-1) to (Cdots);
\draw [-slim] (Cdots) to (C_2);

\node[vertex,label=above left:$R_1$] (R_1)   at (0,7.5)  {};
\draw [-slim] (C_1) to (R_1);
\node[vertex,label=above left:$R_2$] (R_2)   at (2,7.5)  {};
\draw [-slim] (C_2) to (R_2);
\node (Rdots)                              at (4,7.5)  {\dots};
\node[vertex,label=above left:$R_{n-1}$] (R_n-1) at (6,7.5)  {};
\draw [-slim] (C_n-1) to (R_n-1);
\node[vertex,label=above left:$R-n$] (R_n)   at (8,7.5)  {};
\draw [-slim] (C_n)    to (R_n);

\node[vertex,label=above left:$a_1$] (A_1)   at (0,5) {};
\node[vertex,label=above left:$a_2$] (A_2)   at (2,5) {};
\node (Adots)                              at (4,5) {\dots};
\node[vertex,label=above left:$a_{n-1}$] (A_n-1) at (6,5) {};
\node[vertex,label=above left:$a_n$] (A_n)   at (8,5) {};

\draw [-slim] (A_2)   to (A_1);
\draw [-slim] (Adots) to (A_2);
\draw [-slim] (A_n-1) to (Adots);
\draw [-slim] (A_n)   to (A_n-1);

\draw [-slim,bend left] (A_n)      to (A_2);
\draw [-slim,bend left=40] (A_n)   to (A_1);
\draw [-slim,bend left=20] (A_n-1) to (A_2);
\draw [-slim,bend left] (A_n-1)    to (A_1);

\draw [-slim,rounded corners] (A_1)   to (0,3) to (10,3) to (10,10) to (C_n);
\draw [-slim,rounded corners] (A_2)   to (2,3) to (10,3) to (10,10) to (C_n);
\draw [-slim,rounded corners] (A_n-1) to (6,3) to (10,3) to (10,10) to (C_n);
\draw [-slim,rounded corners] (A_n)   to (8,3) to (10,3) to (10,10) to (C_n);

\draw [-slim] (R_1)   to (A_1);
\draw [-slim] (R_2)   to (A_2);
\draw [-slim] (R_n-1) to (A_n-1);
\draw [-slim] (R_n)   to (A_n);

\end{tikzpicture}
\end{center}
\caption{Strategy~$\sigma_\lc$ uses an unbounded number of additional cops.}
\label{fig_against_leaving_cops}
\end{figure}
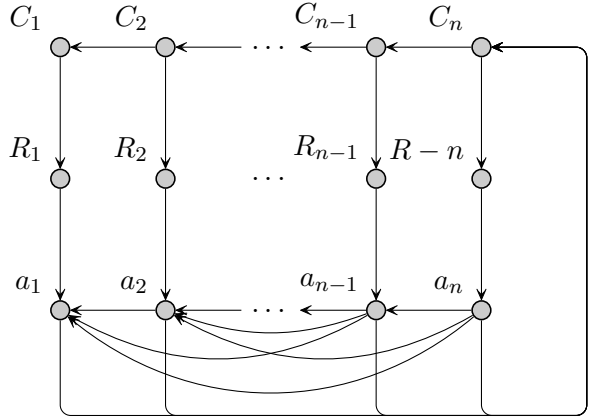

\begin{lemma}\label{lemma_against_leaving_cops}
Let $m\ge 1$. There is a class of graphs $(G_n)_{n>0}$ and winning 
strategies~$\sigma^n$ for $4m$ cops such that $\sigma^n_\lc$ uses 
$4m^2(n-1)$ cops.
\end{lemma}
\begin{proof}
Consider the graph $G_n$ in Figure~\ref{fig_against_leaving_cops}.
It consists of $n$ cliques $C_1,\dots,C_n$ of size $2m$,
$n$ cliques $R_1,\dots,R_n$ of size $3m$, and $n$ independent sets
$A_1,\dots,A_n$ of size $m$ (we could also take cliques instead of
independent sets). Each clique $C_i$ is connected to the clique $C_{i-1}$,
for $i=2,\dots,n$, \ie there are edges from every vertex of $C_i$ to
every vertex of $C_{i-1}$. Furthermore, each $C_i$ is connected to $R_i$, each
$R_i$ to $A_i$, and each $A_i$ to $C_n$ for $i=1,\dots,n$.
Finally, each $A_i$ is connected to $A_j$ for $i=2,\dots,n$ and $j<i$.
    
The strategy $\sigma^n$ is as follows. At the beginning, $2m$ cops occupy $C_1$.
We can assume that the robber goes to $R_1$ because all other components are
reachable from $R_1$. Then $m$ cops occupy $A_1$. If the robber remains
in $A_1$, the cops from $C_1$ go to $R_1$ and capture the robber.
The other possibility for the robber is to switch to the component that
contains $C_2$. Now the cops from $A_1$ are taken away from the graph
(inducing weak non-monotonicity). The robber can only remain in his component.
Then $2m$ new cops occupy $C_2$, the robber goes to $R_2$, $m$ cops from $C_1$
occupy $A_2$, the robber switches to the component containing~$C_3$,
the cops are taken from $A_2$ and the rest of $C_1$ and placed on $C_3$
and so on. In the last step, the cops occupy $C_n$, the robber is in $R_n$
and $m$ cops occupy $A_a$. The robber switches to some $A_i$, but the cops
from $A_n$ expel him from any $A_i$ and the robber is captured in $A_1$.
Note that in every move during the described game, $A_n$ and thus all $A_i$
are reachable from the robber component. Hence, the robber sticking to
the same strategy as above (always switching to the new $C_i$),
the strategy $\sigma^n_\lc$ prescribes to leave cops on all $A_i$.
\end{proof}

\subsubsection{Freezing the context is not enough} \label{subsec:lgo}

Given a strategy~$\sigma$, the \emph{context freezing} strategy~$\sigma^\freeze$ is
obtained from~$\sigma$ as follows. We define two memory variables:~$P$ stores the
placement of cops as if we played according to~$\sigma$ (analogously
to the case of~$\sigma_\lc$) and $\calR = (R_1,\ldots,R_n)$ is a stack of memorized
robber components with $R_{i+1} \subset R_i$, for all~$i$.  Initially,
$P=\0$ and $\calR = ()$ is the empty stack. A robber move $(C,C',v)
\to (C',w)$ does not change~$P$ and~$\calR = (R_1\ldots,R_i)$ is updated by deleting all
$R_j$ with $w\notin R_j$. 

For the cop move, let $(C,v)$ be a position in a play
consistent with~$\sigma^\freeze$ played so far and let~$P$ and~$\calR
= (R_1,\ldots,R_i)$ be the current memory state. The variable~$P$ is
updated to $\sigma(P,v)$. We define $\sigma^\freeze$ by 
\[\sigma^\freeze(C,v) = \bigl(C\setminus \flap_G(v,C)\bigr)\cup
\bigl(\sigma(P,v)\cap \flap_G(v,C)\bigr)\,,\]
\ie ``the context'' ($G - \flap_G(v,C)$) is not changed and we
place cops as prescribed by~$\sigma$, but only within the robber component.
If $(\sigma(P,v) \cap R_i) \setminus \flap_G(v,C) = \0$ (all cops are
placed outside the robber component), the stack~$\calR$
remains unchanged. Otherwise we push $\flap_G(v,C)$ on the stack,
thus freezing the new context.

The next lemma states that changing an arbitrary weakly monotone winning
strategy~$\sigma$ to~$\sigma^\freeze$ may introduce an unbounded
number of additional cops: it is essential that the cops are placed
also in the context. The counter examples are double trees,
shown in Figure~\ref{pic_offhanded}. 

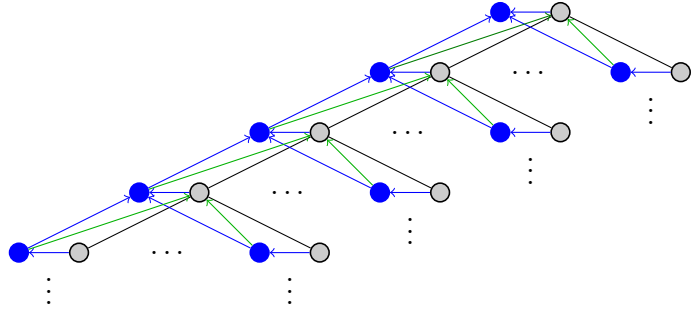
\begin{figure}
\begin{center}
\begin{tikzpicture}[scale=0.8]

\node[vertex] 		(epsilon)	 at (0,0) {};
\node[vertex,blue] 	(epsilon') at (-1,0) {};
\node[vertex] 		(1) at (-2,-1) {};
\node[vertex,blue] 	(1') at (-3,-1) {};
\node[vertex] 		(11) at (-4,-2) {};
\node[vertex,blue] 	(11') at (-5,-2) {};
\node[vertex] 		(111) at (-6,-3) {};
\node[vertex,blue]	(111') at (-7,-3) {};
\node[vertex] 		(1111) at (-8,-4) {};
\node[vertex,blue] 	(1111') at (-9,-4) {};
\node[vertex] 		(n) at (2,-1) {};
\node[vertex,blue] 	(n') at (1,-1) {};
\node[vertex] 		(1n) at (0,-2) {};
\node[vertex,blue] 	(1n') at (-1,-2) {};
\node[vertex] 		(11n) at (-2,-3) {};
\node[vertex,blue] 	(11n') at (-3,-3) {};
\node[vertex] 		(111n) at (-4,-4) {};
\node[vertex,blue] 	(111n') at (-5,-4) {};

\draw (epsilon) to (1);
\draw (epsilon) to (n);
\draw (1) to (11);
\draw (1) to (1n);
\draw (11) to (111);
\draw (11) to (11n);
\draw (111) to (1111);
\draw (111) to (111n);

\draw[<-,blue] (epsilon') to (1');
\draw[<-,blue] (epsilon') to (n');
\draw[<-,blue] (1') to (11');
\draw[<-,blue] (1') to (1n');
\draw[<-,blue] (11') to (111');
\draw[<-,blue] (11') to (11n');
\draw[<-,blue] (111') to (1111');
\draw[<-,blue] (111') to (111n');

\draw[<-,blue] (epsilon') to (epsilon);
\draw[<-,blue] (1') to (1);
\draw[<-,blue] (11') to (11);
\draw[<-,blue] (111') to (111);
\draw[<-,blue] (1111') to (1111);

\draw[<-,blue] (n') to (n);
\draw[<-,blue] (1n') to (1n);
\draw[<-,blue] (11n') to (11n);
\draw[<-,blue] (111n') to (111n);

\draw[->,green!70!black!70!black] (1') to (epsilon);
\draw[->,green!70!black] (n') to (epsilon);
\draw[->,green!70!black] (11') to (1);
\draw[->,green!70!black] (1n') to (1);
\draw[->,green!70!black] (111') to (11);
\draw[->,green!70!black] (11n') to (11);
\draw[->,green!70!black] (1111') to (111);
\draw[->,green!70!black] (111n') to (111);

\node (dots) at (-0.5,-1){\dots};
\node (dots) at (-2.5,-2){\dots};
\node (dots) at (-4.5,-3){\dots};
\node (dots) at (-6.5,-4){\dots};

\node (dots) at (1.5,-1.5){$\vdots$};
\node (dots) at (-0.5,-2.5){$\vdots$};
\node (dots) at (-2.5,-3.5){$\vdots$};
\node (dots) at (-4.5,-4.5){$\vdots$};
\node (dots) at (-8.5,-4.5){$\vdots$};

\end{tikzpicture}
\end{center}
\caption{$\DAGW(G_n)=4$ but the robber wins against $n$ cops if they never guard.}
\label{pic_offhanded}
\end{figure}

\begin{lemma}[\cite{PuchalaRab12}] \label{lemma:OffhandedCops}
There exist graphs $G_n$, for all $n \in \bbN$, such that $\DAGW(G_n) \leq 4$
but every winning strategy of the cops which is restricted to place cops
only inside the robber component uses at least $n+1$ cops.
\end{lemma}

\begin{proof}[Sketch]
Let, for $i\in\{0,1\}$ and $0<m,n\in \bbN$,
$ A(i,m,n) = (\{1, \dots, n\}\times \{i\})^{\le m}$ be the set of
all sequences of length at most $m$ over the alphabet $\{1, \dots, n\}$ labeled 
with $i$ (the labeling is used only to distinguish vertices).
Let, for a $v = (v_0,i),\dots,(v_l,i)\in A(i,m,n)$, $v'$ be the word
$(v_0,1-i),\dots,(v_l,1-i) \in A(1-i,m,n)$. Consider the following class
of directed graphs $G_n = (V_n, E_n)$ for $0< n\in\bbN$
(see Figure~\ref{pic_offhanded}). Hereby $V_n = T^0_n\cup T^1_n$ where 
$T^0_n = A(0,m+1,n)$ and $T^1_n = A(1,m+1,n)$.
The edges are defined by $E_n = E^0_n \cup E^1_n\cup E'_n$. Hereby 
$E^0_n = \{(v, vj),  (vj,v) \mid v \in A(0,n,n), j \in A(0,1,n)\}$, 
$E^1_n = \{(vj, v) \mid v \in A(1,n,n), j \in A(1,1,n)\}$, and
$E'_n = \{(v,v') \mid v\in A(0,m+1,n)\} \cup 
  \{(vj, v') \mid v \in A(1,n,n), j \in A(1,1,n)\}$.

The first statement of the theorem is easy to see. For the second one,
note that it makes no sense for the cops to leave out holes, \ie to place cops
on subtrees of $(T^0_n,E^0_n)$ rooted at a vertex~$v\in T^0_n$, but not on $v$. 
Indeed, due to the high branching degree,  the robber can switch between
subtrees of $v$ going into those having no cop in them until $v$ is occupied by
a cop. Clearly, in this position, there is no need to have cops in subtrees
other than the one with robbers in it. So we can assume that the cops play
top-down, \ie they never leave out holes. Then the robber strategy is just to
stay in the left-most branch. Note that after a vertex $v\in T^0_n$ is occupied
by a cop the vertex $v'\in T^1_n$ is not in the SCC of the robber any more.
It is easy to see that more and more cops become bounded, \ie for every cop
on a vertex $v$, there is a cop free path from the robber vertex to $v$.
\end{proof}


    \subsubsection{Combining leaving cops and freezing}\label{sec_weak_to_strong}

In this section we define a translation of a winning strategy in the
weakly monotone game to a winning strategy in the strongly monotone
game. Our solution is a combination of the approaches discussed in the
previous sections. First we describe the strategy informally.

Recall that $\front_G(R,X)$ is the inclusion minimal
subset of~$X$ that blocks $R \to X$ in~$G$.
The combined strategy~$\sigma_\lc^\freeze$ is obtained from~$\sigma$
as follows. Recall that robber components are defined with respect to
the set~$M$ of vertices that have been occupied by chasers. The cops start to play sticking to~$\sigma$ until the
robber changes his component. If~$\sigma$ prescribes to remove a cop
from some vertex~$v$ such that weak non-{}monotonicity occurs, the cop
is not removed (and neither any other tied cops, but let us concentrate on~$v$ for now).  The cops play further according
to~$\sigma$ as if the tied cops were removed until the robber chooses
a component~$R$ such that $v\notin R$. If that never happens and the
robber always chooses the component containing~$v$, then a new
cop, say occupying~$w$, can be tied only if the cop on~$v$ becomes not tied (and we can
continue with~$w$ in place of~$v$). Indeed, if~$\sigma$ is strongly
monotone against a shy robber, a cop can become tied only
if the (now not shy) robber changes his component; if the robber component~$R$
contains~$v$ and the robber leaves~$R$ towards some~$R'$, then~$v$ is
not reachable from~$R'$ by Lemma~\ref{lemma_no_left_to_right}.

When we have $v \notin R$, the context of~$R$ is frozen. In that
position we have at most~$k$ cops in the context of~$R$
including~$v$. Now the cops play according to~$\sigma$ restricted to~$R$
until the robber leaves it, or he is captured. Hereby, if a cop becomes unreachable
from the robber vertex, he is removed from the graph and can be reused
later. If the robber leaves~$R$ and enters another component~$R'$, the
placement of the cops outside~$R$  is still the same as when the
robber chose~$R$ (and not~$R'$), so he will be captured in or expelled from~$R'$ and
from any other such component in the same way as for~$R$, and the cops win.

It remains to see that the new strategy uses at most~$k^2$ cops. Our
argument is that any tied cop becomes untied before $k^2$
other cops become tied. When a cop on~$v$ is tied, we freeze at
most~$k$ cops and continue to play only within~$R$. Up to the
position when according to~$\sigma$ the last~$k$th cop enters~$R$,
we have enough cops by \emph{induction on the number of cops used
by~$\sigma$:} we need at most $(k-1)^2 = k^2-(2k-1)$ cops. If the robber is already captured in~$R$, we are
done. If he leaves~$R$ before the $k$th cops enters~$R$ according
to~$\sigma$, we argue for his next component as for~$R$. For the (most
interesting) case that the $k$th cop comes to~$R$, we are going to show that
the cop tied at the beginning, \ie that on~$v$, is now untied.
If the cop on~$v$ was still tied, he would be in particular reachable from
the robber vertex, so the robber can leave~$R$. We will see that then he can also
reach some vertices that induce strong non-{}monotonicity, which
contradicts the fact that~$\sigma$ is weakly monotone.

\paragraph{The combined strategy}

Let~$\sigma$ be a shy-{}similar strategy. We define the combined
strategy~$\sigmacomb$ by induction on the maximal number of cops that
appear in a play consistent with~$\sigma$. The combined strategy uses
the same memory as~$\sigma^\freeze$: it keeps track of a play
consistent with~$\sigma$ by memorizing a cop placement~$P$ and stores
the history of freezing in a stack~$\calR$. If~$p$ is a position, then
$P(p)$ and $\calR(p)$ are the values of~$P$ and~$\calR$, respectively,
in position~$p$. At the beginning of a play, $P(\0) = \0$ and $\calR(\0)=()$. A
robber move $(C,C',v) \to (C',w)$ does not change~$P$
and~$\calR = (R_1\ldots,R_i)$ is updated by deleting all $R_j$ with
$w\notin R_j$.

For the cop move, let $(C,v)$ be a position in a play
consistent with~$\sigmacomb$ played so far and let~$\calR(C,v)= (R_1,\ldots,R_i)$.
 The variable~$P$ is updated to $\sigma(P,v)$. In the cop move, there are two differences
to~$\sigma^\freeze$. The more substantial one is that now the cops are
placed also outside the robber component, but still not outside~$R_i$. The
other one is that we can change the cop placement in the context
removing some cops if this does not directly lead to
non-{}monotonicity. Note that removing those cops could also be
performed for~$\sigma^\freeze$. We did not do it to keep the
description of~$\sigma^\freeze$ simpler, but this would not
make~$\sigma^\freeze$ work. Formally we define
\[\sigmacomb(C,v) \coloneqq \front_G(v,C) \cup (\sigma(P(C,v),v) \cap R_i)\,.\]
Note that within~$R_i$, even~$\sigma$ prescribes only strongly
monotone moves and removes unreachable cops, so we could represent~$\sigmacomb$ 
in a way more similar to~$\sigma^\freeze$:
\[\sigmacomb(C,v) = (\hat C \setminus R_i) \cup (\sigma(P(C,v),v) \cap R_i)\]
where $\hat C = C \cap \front_G(v,C)$.

We update the stack~$\calR$ by pushing $\flap_G(v,C)$ on~$\calR$
if $(\sigma(P(C,v),v) \cap R_i) \setminus \flap_G(v,C) \neq \0$ and
let~$\calR$ unchanged otherwise.

As $\front_G(v,C) \subseteq \sigmacomb(C,v)$, it is immediately
clear that~$\sigmacomb$ is strongly monotone. It also guarantees a capture
of the robber because~$\sigma^\freeze$ does it and
$\sigma^\freeze(C,v) \subseteq \sigmacomb(C,v)$, so~$\sigmacomb$ 
prescribes to place a cop into the robber component again
and again. We have to prove that~$\sigmacomb$ uses at most~$k^2$ cops.

\begin{lemma}
If $\sigma$ is a shy-similar winning strategy for~$k$ cops, then $\sigmacomb$ uses at
most~$k^2$ cops.
\end{lemma}
\begin{proof}
We show that, for every shy-similar winning strategy~$\sigma$ for~$k$ cops in the
weakly monotone game and every graph~$G$, every tied cop in a play
consistent with~$\sigmacomb$ on~$G$ becomes not reachable from the robber
vertex before $k^2$ new cops are tied. Clearly, this implies the
statement of the lemma. 

The proof is done by induction on~$k$. Without loss of generality,
assume that~$G$ is strongly connected (otherwise repeat the
argument for every strongly connected component).
Consider a fixed play~$\pi$ consistent with~$\sigmacomb$. If~$k=1$,
then there are no tied cops and the statement is trivial. Let $k>1$.
Let $(C_0,C_1,v_0)\to(C_1,v_1)$ be a move in~$\pi$ such that 
$\calR_0 = \calR(C_0,v_0)=(R_1,\ldots,R_i)$ and $\calR_1 =
\calR(C_1,v_1) = (R_1,\ldots,R_j)$ with $j\le i$. 
Assume that the move results in a new tied cop on~$v\in C_1$. 

Either~$v$ is in the robber component until the end of~$\pi$, or
not. We show that in the first case there are no further tied cops, so
we are done. As~$\sigma$ is shy-{}similar, tied cops appear only when
robber changes his component, compare the proof of
Proposition~\ref{prop_strong_shy_weak}. If another cop on~$w$ becomes
tied, then the robber changes his component, say from~$R_1$
to~$R_2$. But then either~$v\notin R_1$, or~$v\notin R_2$, so there is
a position in that the robber is not in the same component as~$v$. 

For the other case, let~$(C_2,C_3,v_2)\to(C_3,v_3)$ be the first move
such that $v\notin \flap(v_3,C_3)$. Let $\calR(C_3,v_3) = (R_0,\ldots,R_m)$.
Without loss of generality we can assume that the rest of~$\pi$ is played in~$R_m$.
Otherwise, until the robber leaves~$R_m$ and goes to some
component~$R'$, cops are placed only in~$R_m$, which is not reachable
from~$R'$ without introducing strong non-{}monotonicity (by
Lemma~\ref{lemma_no_left_to_right}), so we would repeat our
arguments for~$R_m$ for~$R'$.

Until~$\sigma$ prescribes to place~$k$ cops in~$R_m$, \ie while
$|\sigma(P,v)\cap R_m| < k$, by the induction hypothesis for~$\sigma$ and~$R_m$, we have at
most $(k-1)^2$ tied cops. Note that~$R_m$ is strongly connected, so we
do not violate our assumption that the graph on that we play is
strongly connected. Consider the first move $(C_4,v_4) \to
(C_4,C_5,v_4)$ with $|P(C_5,v_4) \cap R_m| =k$. We want to show that
then~$v$ if not reachable from~$v_4$, which means that the cop on~$v$
is not tied any more.

As~$G$ is strongly connected, there is a path~$P$ from~$v$
to~$v_4$ in~$G$. Recall that $v\notin R_m$ by the case distinction, but by our assumption
that the remaining of the play takes place in~$R_4$, we have $v_4\in
R_m$. If~$v$ were reachable from~$R_m$ in $G - P(C_5,v_4)$, by
Lemma~\ref{lemma_no_left_to_right}, cops outside of~$R_m$
block~$\{v_4\}\to R_m$. However, in position $(P(C_5,v_4),v_4)$ all cops are
in~$R_4$, a contradiction.
\end{proof}

We can count the number of additional cops more accurately. For the
first tied cop we need to freeze at most~$k$ new cops, for the next
tied cop at most~$k-1$ cops and so on, so in total, we can come up
with $k^2/2 + k/2$ cops. Finally, we obtain the desired result.

\begin{theorem}\label{thm_main}
If $k$ cops have a winning strategy in the weakly monotone shy robber game,
then $18k^2+3k$ cops have a winning strategy in the strongly monotone game.
\end{theorem}
\begin{proof}
  Assume that $k$ cops have a winning strategy $\sigma$ in the weakly
  monotone shy robber game. Then by Corollary~\ref{cor_compare_games}
  $\mb(\sigma)$ is a winning strategy for $3k$ cops in the strongly
  monotone shy robber game. By Proposition~\ref{prop_strong_shy_weak}
  one needs $2\cdot 3k = 6k$ cops to win weakly monotonically. Finally
  one needs $((6k)^2+6k)/2 = 18k^2+3k$ cops to win strongly monotonically.
\end{proof}


\section{Comparing Width Measures with Respect to Generality}\label{sec:ordering}

    This section is devoted to the question, given two measures~$a$
and~$b$, whether the class of graphs with bounded values of~$a$ is a
subclass of the class of graphs with bounded values of~$b$.

\subsection{Comparing \Dagw and \Kw}\label{sec_cmp_Kelly}

\Kw is a complexity measure for directed graphs introduced by Hunter 
and Kreutzer in~\cite{HunterKre08}. \Kw is similar to \dagw and can be
defined by a decomposition, by a graph searching game and by an
elimination oder, similar to \tw. 

An \emph{elimination order}~$\elorder$ for a graph $G =
(V,E)$ is a linear order on~$V$. For a vertex~$v$ define $V_{\elorderg
  v} \coloneqq \{u\in V \mid v\elorder u\}$. The
\emph{support} of a vertex~$v$ with respect to~$\elorder$ is
\[\supp_{\elorder}(v) \coloneqq \{u\in V \mid v\elorder u \text{ and there is }
v'\in\Reach_{G-V_{\elorderg v}}(v)\text{ with }(v',u)\in E\}\,.\] The
\emph{width} of an elimination order~$\elorder$ is $\max_{v\in
  V}|\supp_{\elorder}(v)|$. The \emph{\kw} $\KW(G)$ of~$G$ is one plus the
minimum width of an elimination order of~$G$.

Hunter and Kreutzer conjecture
in~\cite[Conjecture~$30$]{HunterKre08} that \dagw and \kw bound each
other by a constant factor. More generally, the question is whether
there is a function~$f\colon \bbN\to\bbN$ such that, for every
graph~$G$, we have 
\begin{enumerate}[(1)]
\item $\DAGW(G) \le f\cdot\KW(G)$ and
\item $\KW(G) \le f\cdot \DAGW(G)$. 
\end{enumerate}
In~\cite{HunterKre08} it is shown
that if $\KW(G) = k$, then $2k-1$ cops have a (possibility
non-{}monotone) winning strategy~$\sigma$ in the \dagw game. We demonstrate
that~$\sigma$ is, in fact, weakly monotone, thus answering the first
question affirmatively.

\begin{theorem}\label{thm_kw_to_dagw}
If $\KW(G) = k+1$, then $\wmDAGW(G) \le 2k+1$.
\end{theorem}
\begin{proof}
Our proof follows the proof of Theorem~$20$
from~\cite{HunterKre08}, which shows that an elimination order of
width~$k$ induces a (possibly non-{}monotone) strategy for $2k+1$ cops
in the \dagw game. What we prove additionally is just that the constructed
strategy is weakly monotone.

Let $\elorder$ be an elimination order for~$G$ of width~$k$. We
define a weakly monotone winning strategy~$\sigma$ for $2k+1$ cops in
the weakly monotone game on~$G$.

Any play consistent with~$\sigma$ can be partitioned into two kinds of rounds:
the blocking rounds and the chasing rounds. A blocking round consists
of a blocking cop move and an answer of the robber. A chasing round may
contain a longer sequence of moves. The cops are divided into two
teams: a team of~$k+1$ blockers and a team of $k$ chasers. While a play proceeds, a cop
may change his team. 
 
During the play, after every blocking round, the following invariant
will hold. Let the blockers occupy the
  set of vertices~$B$, let the chasers occupy the set~$C$ and let the
  robber be on~$v$.
\begin{enumerate}
\item $|B| \le k+1$, $C = \0$.

\item If~$u$ is the $\elorder$-least vertex from~$B$, then $v\elorder
  u$ and~$B$ blocks $\{v\} \to V_{\elorderg u}\setminus B$.
\end{enumerate}
In the first move, $k+1$ blockers occupy the $\elorder$-maximal vertices
of~$G$ and the robber chooses some vertex. It is trivial that the
invariant holds.

Consider a position after some blocking round has been just
finished. Let~$v$ be the robber vertex,~$B$ the set of vertices
occupied by blockers and~$C$ the set of vertices occupied by chasers
such that the invariant holds. Let~$u = \min_{\elorder} (B)$ and let
$\calR$ be the set of components of~$G - V_{\elordergeq u}$ where
$V_{\elordergeq v} = V_{\elorderg} \cup \{u\}$. Let~$\setorder$ be
the linear order on~$\calR$ defined by $R\setorder R'$ if and only if
$\max_{\elorder}(R) \elorder \max_{\elorder}(R')$ and let
$\setordereq$ be its reflexive closure. Let~$R$ be the component
in~$\calR$ with $v\in R$ and let $w = \max_{\elorder}(R)$. The chasing
round proceeds as follows. The cops announce to place (at most~$k$)
chasers on $\supp_{\elorder}(w)$. The robber choses a vertex $v'$ in a
component $R'$. If $R' \setordereq R$, then in the next position the
chasers on $\supp_{\elorder}(w)$ block every path from~$v'$ to
$V_{\elordergeq u'}$ where $u' =
\min_{\elorder}(\supp_{\elorder}(w))$. (Indeed, assume that there is a
path~$P$ from~$v'$ to some~$u''$ with $u'\elorder u''$ such that
$P\cap \supp_{\elorder} (w) = \0$. Let $(a,b) $ be the first edge with
$u'\elorder b$. Then there is a path from~$w$ to~$b'$ via~$v'$, thus
$b\in \supp_{\elorder} (w)$, a contradiction.)  This competes the
chasing round.

If $R\setorder R'$, then the chasers are
removed from $\supp_{\elorder} (w)$ and placed on $\supp_{\elorder}(w')$
where $w' = \max_{\elorder}(R')$. As the robber can change to a
$\setorder$-greater component only until he reaches~$B$, this process
is finite and at some point the robber is blocked by the chasers, \ie
we have the previous case. Note that by the definition of
$\supp_{\elorder}$, for all $w\in V$ we have $w\elorder
\min_{\elorder}(supp_{\elorder}(w))$. Hence, as $w$ is chosen to be the
$\elorder$-maximal in the robber component, the chasers are always placed
outside of the robber component. Thus placing and removing them in a
chasing round never induces strong non-{}monotonicity.

When the chasing round is over, the next blocking round begins. Let
$v$ be the robber vertex, $R$ the robber component of $G -
V_{\elorder u}$ (where $V_{\elorder  u} = \{v\in V \mid v\elorder u\}$
and $u$ is still the $\elorder$-minimal element
in~$B$), and~$w$ the $\elorder$-maximal element of~$R$. In the
blocking round, the chasers from the previous chasing round become
blockers (let~$B'$ be the set of vertices they occupy) and a blocker
from~$B$ is placed on~$w$. Other old blockers from~$B$ become chasers
and are removed from the graph. After the robber makes his move (say,
he goes to~$v'$), the blocking round is finished.

We have to check that the invariant still holds and that no strong
non-{}monotonicity occurred during the last blocking round. The first
invariant property holds because we used at most~$k$ cops as chasers
and the new blockers are the old chasers plus the cop
on~$w$. Furthermore, the chasers have been removed from the graph. The second
property (that $B'$ blocks $\{v'\} \to V_{\elorderg u'}\setminus B$
where $u' = \min_{\elorder}(B')$) holds by the construction and
implies that removing cops from~$B$ was (even strongly) monotone.
Finally, the space available for the robber shrinks after every
blocking round because the cops occupy~$w$, so the robber is finally captured. 
\end{proof}

\begin{corollary}
  If $\KW(G) = k$, then $\DAGW(G) = \calO(k^2)$.
\end{corollary}

    \subsection{Separating D-width from DAG-width, Kelly-width and \dtw}

Safari suggests in~\cite{Safari05} \dw as another structural
complexity measure. Recall that for a directed graph~$G$, we denote
its undirected underlying grap by~$\bar G$. A \emph{D-decomposition}
of a graph~$G$ is a pair $(T,(X_t)_{t\in V(T)})$ where~$T$ is a
directed tree with edges oriented away from the root. Furthermore for
all $t\in V(T)$, $X_t \subseteq V(G)$ and
\begin{enumerate}[(1)]
  \item $\bigcup_{t\in V(T)}X_t = V(G)$, and
  \item for all strongly connected sets $S \subseteq V(G)$
    the underlying undirected subgraph of $T[\{t\in V(T) \mid X_t\cap
    S \neq \0\}]$ is a connected subtree of $\bar T$.
\end{enumerate}
The \emph{width} of $(T,(X_t))$ is $\max_{t\in
  V(T)}|X_t|$.~\footnote{In~\cite{Safari05} the width is $\max_{t\in
    V(T)} |X_t|-1$.} The \emph{\dw} $\DW(G)$ of~$G$ is the minimum width of a
D-decomposition of~$G$.

The following definition of \dw may be more useful in algorithmic
applications and suits our goals better. For a graph~$G$, if $X, Y\subseteq V(G)$
and~$X$ is a union of strongly connected components of $G-Y$, we say
that~$X$ is \emph{$Y$-normal.}  DS-width is a complexity measure that
differs from \dw at most by the factor of two.  Let $G$ be a graph. A
\emph{DS-decomposition} of $G$ is pair $(T,(X_t)_{t\in V(T)})$ where
$T$ is a directed tree with edges oriented away from the root $r$ and
$X_t$ are sets of vertices of $G$ such that the following holds. Let
$X_{\ge t} = \bigcup_{q\ge t}X_q$ for all $t\in V(T)$. Then
  \begin{enumerate}[(1)]
  \item $\bigcup_{t\in V(T)}X_t = V(G)$,
  \item for all $v\in V(G)$ the set $\{t\in V(T) \mid v\in X_t\}$ is
    connected in $\bar T$,
  \item for all edges $(s,t)\in E(T)$, $X_{\ge t} \setminus X_s$ is $(X_s\cap 
X_t)$-normal.
  \end{enumerate}
The width of $(T,(X_t))$ is $\max_{t\in V(T)}
|X_t|$. The \emph{DS-width}
of $G$, $\DS(G)$ is the minimum of the widths of all DS-decompositions of~$G$.


\begin{lemma}[See~\cite{Gruber08}]\label{lemma_cmp_dw_dprimw}
For all graphs $G$, $\DW(G)\le \DS(G)\le 2\DW(G)$.
\end{lemma}
\begin{proof}
  Let $(T,(X_t)_{t\in V(T)})$ be a D-decomposition of width $k$. We
  obtain a DS-decomposition $(T',(X'_{t'})_{t'\in V(T')})$ of width
  $2k$ from $(T,(X_t)_{t\in V(T')})$ as follows. Replace every edge
  $(s,t)\in E(T)$ by a new node $st$ and edges $(s,st)$ and $(st,t)$
  and let $X_{st}$ be $X_s\cup X_t$. Then $(T',(X_{t'}')_{t'\in
    V(T')})$ is a DS-decomposition of~$G$. Indeed, assume that for
  some edge $(s,st) \in E(T')$ there is a path~$P$ from some $w\in
  X_{\ge st}\setminus X_s = X_{\ge t}\setminus X_s$ to some $v\in
  V(G)\setminus (X_{\ge st}\setminus X_s)$ and back to $w$. Assume for
  a contradiction that~$P$ avoids $X_s \cap X_{st} = X_s$. Then~$P$ is
  a strongly connected subgraph of~$G$, so by the second property of
  D-decomposition, the set $\{t\in V(T) \mid P\cap X_t \neq \0\}$ is
  connected in~$\bar T$. Thus $X_s \cap P \neq \0$,
  but we assumed that this is not true. Hence $X_{\ge st}\setminus
  X_s$ is $(X_s \cap X_{\ge st})$-normal.  With the same argument one
  can see that for all edges of the form $(st,t)$, the set $X_{\ge t}
  \setminus X_{st}$ is $(X_t \cap X_{\ge st})$-normal.

  Properties (1) and (2) follow trivially form the properties of the
  D-de\-com\-po\-si\-tion. Furthermore, it is clear that the width of
  $(T',(X_{t'})_{t'\in V(T')})$ is at most~$2k$.

  Now assume that we have a DS-decomposition $(T,(X_t)_{t\in V(T)})$
  of width~$k$. We show that it is also a D-decomposition. Let $S$ be
  a strongly connected set of~$G$ and assume that $\{t\in V(T) \mid
  X_t \cap S \neq \0\}$ is not connected in~$\bar T$. Let~$q$,~$s$
  and~$t$ be some nodes of~$T$ such that $X_q \cap S \neq \0$, $X_t
  \cap S \neq \0$, $X_s = \0$ and~$s$ is on the path $P_{tqs}$
  from~$t$ to~$q$ in $\bar T$. Choose $s$, $q$ and $t$ such that
  $P_{tqs}$ has minimal length. Either $q<s$ or $q<t$, say $q<t$, then
  $(q,t)\in E(T)$ (the case $q<s$ is analogous). As~$S$ is
  strongly connected, there is a path~$P$ from $X_{\ge t} \setminus
  X_q$ to $X_{\ge s}\setminus X_q$ and back within~$S$. As $S\cap X_q
  = \0$, $P$ avoids $X_q$. Note that $(X_{\ge s} \setminus X_q) \cap
  (X_{\ge t} \setminus X_q) = \0$, so~$P$ leaves $X_{\ge t}\setminus
  X_q$ and returns there without visiting $X_q$ thus violating the
  normality condition of the DS-decomposition for the edge $(q,t)$.
\end{proof}

We separate \dw from \dtw, \dagw, \kw, and from the cop- and
robber-monotone component game. First we show in
Theorem~\ref{thm_dw_game} that if \dw is bounded, then a bounded
number of cops suffices to capture the robber in a cop- and
robber-monotone way in the component game. It follows that then \dtw
is bounded as well (this is already known from~\cite{Safari05}). It is
known that there are classes of graphs where \dtw is bounded but
neither \kw, nor \dagw are: undirected binary trees with additional
edges forming the upward transitive
closure~\cite{BerwangerDawHunKreObd12}. The \dw of those graphs is
also bounded. We show that there is a class
$\calG$ of graphs where three (four) cops win in the cop- and
robber-monotone component (resp.\@ reachability) game, but whose \dw
is unbounded (Theorem~\ref{thm_dw_game_fails}). Hence, \dtw and \dagw
are bounded on $\calG$, but \dw is unbounded. We also show that \kw
is bounded on~$\calG$. Finally, we use Theorem~\ref{thm_non_cop_mon}
to separate \dw from \dtw in another way in Theorem~\ref{thm_sep_dw_dtw}.



\begin{theorem}
\label{thm_dw_game}
For all graphs $G$, if there is a DS-decomposition of $G$ of width~$k$,
then~$\cmdTWcops{G}\le k$.
\end{theorem}
\begin{proof}
  Let $(T,(X_t)_{t\in V(T)})$ be a DS-decomposition of width~$k$. The
  cops have the following winning strategy. In te first move they
  occupy $X_r$ where~$r$ is the root of~$T$. In general they keep the
  invariant true that if the current position is $(C,R)$, then $C =
  X_s$ for some $s\in V(T)$ and $R\subseteq X_{\ge t}\setminus X_s$
  for some $t$ with $(s,t)\in E(T)$. Then the next move of the cops is
  to $(C,X_t,R)$ and after the next robber move the invariant
  holds. Note that $X_t \setminus X_s \subseteq R$. Note also that~$R$
  is a strongly connected component of $G - (X_s\cap X_t)$, so the
  play is robber-{}monotone. Furthermore, by property~(2), it is also
  cop-{}monotone. When the cops reach a leaf, the robber is
  captured. Clearly, exactly~$k$ cops are used.
\end{proof}

The opposite direction fails because the cops may be forced to occupy
the same vertex when the robber goes to different components. Assume
that we reached a position $(C,C',R)$ and the robber can choose $R_1'$
or $R_2'$. In both cases after playing some time the cops must occupy
a vertex~$v$: in the first case because $v\in R_1'$ and in the second
case because they have to block the robber in $R_2'$. The
decomposition corresponding that strategy has~$v$ in different
successors of the bag that corresponds to position $(C,R)$, but not in
the bag of $(C,R)$ itself. However this violates the connectivity
condition of a DS-decomposition. Theorem~\ref{thm_dw_game_fails} shows
that the described situation is unavoidable.

In the proof of Theorem~\ref{thm_dw_game_fails} we use a technical
notion of a game which at least partially corresponds to \dw. The \emph{\dw
game} on a graph~$G$ is another type of graph searching games that does not 
match
our framework. At the beginning the cops group components of~$G$ into
equivalence classes and the robber choses one class and goes there. At
this moment the cops do not see the robber and only know his
class. From now on every cop can be placed only within that
class. Then the cops make a move as in all games described before and
group the emerging components within the current class into new
classes and so on.

Formally the cop positions are $(C,\calR)$ where $C\subseteq V(G)$ and
$\calR = \{R_1,\ldots,R_m\}$ is a set of components of $G - C$. Hereby
every $R_i$ is a component the cops consider to be a possible robber
component. The cops can move to a position $(C,C',\sim,\calR)$ where
$C'\subseteq \bigcup_{1\le i \le m}R_i$ and~$\sim$ is an equivalence
relation on components of $G-C'$.  From $(C,C',\sim,\calR)$ the robber
can move to a position $(C',\calR')$ where $\calR' =
\{R'_1,\ldots,R'_s\}$ is the set of components $R'_j$ (for
$j\in\{1,\ldots,s\}$) of $G-C'$ such that $R'_j\subseteq R_i$ for some
$R_i$ and all $R'_j$ are $\sim$-equivalent. In other words the robber
choses an equivalence class of components, a \emph{group.} If there is
a path from some $R'_j$ outside of $R'_j$ and then back to $R'_j$ in
$G- (C\cap C')$, then the robber wins (by the
non-robber-monotonicity). He also wins all infinite and all
non-cop-monotone plays. The cops win if the capture the robber (\ie he
has no legal move).

\begin{lemma}\label{lemma_dw_game}
  If $\DS(G) = k$, then $2k$ cops have a winning strategy in the \dw
  game on~$G$.
\end{lemma}
\begin{proof}
  Let $(T,(X_t)_{t\in V(T)})$ be a DS-decomposition where the root
  of~$T$ is~$r$. Note that for all $(s,t)\in E(T)$ the set
  $\bigcup_{q\ge t}X_q\setminus X_s$ is a union of components of
    $G-X_s$. The cops occupy~$r$ in the first move and for components
    $R_1$ and $R_2$ of $G-X_r$ they define $R_1\sim R_2$ if and only
    if $R_1$ and $R_2$ are both contained in $\bigcup_{q\ge
      t}X_q\setminus X_r$ for some $(r,t)\in E(T)$. The robber choses
    a $\sim$-class, \ie essentially an edge $(r,t)$. Note that the
    robber is blocked in $\bigcup_{q\ge t}X_t\setminus X_r$ by
    $X_r\cap X_t$. Then the
    cops occupy $X_r\cup X_t$ and define the equivalence relation in
    the same way as before whereby now~$t$ plays the role
    of~$r$. Finally the cops capture the robber at the latest in some $X_s$
    for a leaf~$s\in V(T)$. The cop strategy is cop-monotone by the
    monotonicity of the DS-decomposition. 
\end{proof}

\begin{theorem}\label{thm_dw_game_fails}
  There is a class of graphs $G_n$ such that~$3$ cops have a cop-
  and robber-monotone winning strategy in the \dtw and \dagw games
on each $G_n$ and $\KW(G_n) = 4$, but $\DW(G_n) \ge n$.
\end{theorem}
\begin{proof}
  Informally, the graph $G_n$ consists of two vertex disjoint parts
  (below we give a formal definition).  One is a
  copy of $T^n_{n+1}$ with root $r(T^{n^2}_{n+1})$ and additional
  edges forming the downward transitive closure. The other part has
  (another copy of) $V(T^n_{n+1})$ as its vertex set and edges
  $\bigl\{ (va,v) \mid v \in \{0, \ldots, n \}^{n}, a\in
  \{0,\ldots,n\}\bigr\}$. We denote the second part as~$T'$. The parts
  are connected by edges going from a vertex~$v$ in $T^n_{n+1}$ to its
  copy~$v'$ in~$T'$ and from~$v'$ to the parent of~$v$. Now we transform
  the resulting graph by applying the following operation once on each
  vertex of $T^n_{n+1}$ (except the root) in a top-down manner. The
  current vertex $v\in V(T^n_{n+1})$ (except the root) is replaced
  by~$n$ copies $v_1$, \ldots, $v_n$. Let~$w$ be the parent
  of~$v$. Then the edge $\{v,w\}$ is replaced by edges $\{v_i,w\}$ and
  the edge $(v,v')$ by edges $(v_i,v')$. Every $v_n$ is the root of a
  copy of the subtree of $T^n_{n+1}$ rooted at~$v$. Hereby 
  the edges going to and from~$T'$ are also copied.
  
  Formally $V(G_n)$ is the disjoint union of $V(T^{n^2}_{n+1})$ and
  $V(T^n_{n+1})$ with edges $E_T$, $E_{tr}$, $E_{up}$ and $E_\cross$ where
  $E_T = \bigl\{\{v,va\} \mid v\in\{0,\ldots, n^2-1\}^{\le n}, a \in
  \{0,\ldots, n^2-1\} \bigr\}$ are edges forming $T^{n^2}_{n+1}$, $E_{tr}
  = \bigl\{ (v,vw) \mid v,w,vw\in \{0,\ldots,n^2-1\}^{\le n} \bigr\}$ is the
  downward transitive closure on $T^{n^2}_{n+1}$, $E_{up} = \bigl\{ (va,v) \mid
  v\in\{0,\ldots,n\}^{\le n}, a\in \{0,\ldots,n\}\bigr\}$ forms
  $T^n_{n+1}$ and~$E_\cross$ connects $T^{n^2}_{n+1}$ and $T^n_{n+1}$ as follows. For a
  word $w = w_1 \ldots w_{m} \in \{0, \ldots, n^2-1 \}^{\le n}$ with
  $|w|>0$ let $w'$ be the word $w'_1, \ldots, w'_{m}$ with $w'_i =
  \lfloor w_i/n\rfloor$. Then there are edges $(w,w')\in
  V(T^{n^2}_{n+1}\times T^n_{n+1})$ and $(w',w_1\ldots w_{m-1})$ (\ie if $w =
  w_1$, then the edge is $(w',\epsilon)$).

  The winning strategy for the cops in the cop-monotone \dtw game is,
  roughly, to traverse $T^{n^2}_{n+1}$ and $T^n_{n+1}$ in parallel
  downwards. As long the robber is in $T^{n^2}_{n+1}$, the cops
  descend from the root of $T^{n^2}_{n+1}$ to some leaf~$w$ which is
  in the current robber component and in $T^n_{n+1}$ also from the
  root to the leaf~$w'$. If the robber changes to $T^n_{n+1}$, he
  finds himself in a component that consists of one vertex and is
  captured in the next move. It is easy to see that three cops suffice
  to win.

  In the \dagw game, four cops follow the robber, in $T^n_{n+1}$ in
  the same way as in $T^{n^2}_{n+1}$: in parallel downwards.  For the
    \kw, consider the elimination order where every vertex of
    $T^n_{n+1}$ is smaller than every vertex of $T^{n^2}_{n+1}$ and
    vertices within $T^n_{n+1}$ and within $T^{n^2}_{n+1}$ are ordered
    in the obvious way (following the depth-first search). Then the
    support of every vertex is at most three, so the \kw is four.

We show by induction on~$n$ that the robber has a winning strategy in
the D-width game against~$n$ cops. At the beginning he choses the
group containing the component with $r(T^{n^2}_{n+1})$ and continues to do
so as long as the root is not occupied by the cops in a move
$(C,\calR) \to (C,C',\sim,\calR)$. For vertex $v\in V(T^{n^2}_{n+1})$ the
subtree of $T^{n^2}_{n+1}$ rooted at~$v$ is denoted $T_v$ and similarly we
write $T_{v'}$ for the subtree of $T' = T^n_{n+1}$ rooted at $v'$. Let
$r_1,\ldots,r_{n^2}$ be the direct successors of $r(T^{n^2}_{n+1})$.  Due
to the robber-monotonicity the cops have visited vertices in at most
$n-1$ sets $T_{r_i}\cup T_{r'_i}$. Hence there are at least~$n$
vertices $r_i$, say $r_1,\ldots,r_n$, with the same successor~$r'$ in
$T^n_{n+1}$ such that all $T_{r_i}$  for $1\le i\le n$ and $T_{r'}$ are cop free.
Every $T_{r_i}$ and $T_{r'}$ are current components and the
equivalence relation~$\sim$ declared by the cops defines groups of them.

There are two cases. In the first case there are at least two
groups. One of them contains some $T_{r_i}$, but not
$T_{r'}$. The robber choses this group and plays from now on only on
$T_{r_i}$. His strategy is to remain in the group containing a
vertex~$v$ that is possibly high in the tree. Due to the high
branching degree when the cops occupy~$v$, there is a direct
successor~$w$ of~$v$ such that $T_w$ is cop free. The robber choses
the group containing~$w$ and plays further in the same way. During
this play the cops occupying the vertices $v_1,v_2,\ldots$ on the path
from $r(T^{n^2}_{n+1})$ to~$v$ cannot be removed because there is a path
from~$w$ via $T_{r'}$ to all $v_i$ and then back to~$w$. This path is
cop free because the cops are not allowed to occupy $T_{r'}$, which is
not in the current group. Hence if the cops leave some $v_i$, the robber
wins by non-robber-monotonicity. When the robber reaches a leaf,
all~$n$ cops stay om the path from $r(T^{n^2}_{n+1})$ to that leaf and
cannot be removed, so the robber wins.

In the other case there is only one group. The robber may be in each
component~$C$ containing some $T_{r_i}$. Each such component has an
isomorphic copy $G_{n-1}$ as a subgraph. (Note that~$C$ is not a copy
of $G_{n-1}$ because its branching degree is still $n^2$ and not
$(n-1)^2$ as in $G_{n-1}$.) Thus by the induction hypothesis the
robber wins against $n-1$ cops, so if the cops should win, they need
the cop from $r(T^{n^2}_{n+1})$ in~$C$. But then the robber can reach
$r(T^{n^2}_{n+1})$ from another $T_{r_j}$, which causes the
non-robber-monotonicity.
\end{proof}

It was conjectured that \dtw and \dw are the same (\cite[Page
750]{Safari05}\footnote{Safari actually conjectures that D-width
  equals directed tree width which would imply cop-monotonicity.}). We
show, however, that the gap between them is not bounded by any
function (which is clear from
Theorem~\ref{thm_dw_game_fails}, but we give yet another proof.)

\begin{theorem}\label{thm_sep_dw_dtw}
  There is a class of graphs with bounded \dtw and unbounded \dw.
\end{theorem}
\begin{proof}
  Consider the class of graphs from Theorem~\ref{thm_non_cop_mon}. The
  \dtw of the graphs from that class is bounded. If the \dw were
  bounded, DS-width would be bounded as well (Lemma~\ref{lemma_cmp_dw_dprimw}), 
and by Lemma~\ref{thm_dw_game} a bounded number of cops could
  capture the robber on each graph in a cop-monotone way, but this is
  not the case.
\end{proof}

    \subsection{\Otw}\label{subsec:otw}

\RRl{Stephan, wir sollten etwas über die otw schreiben. Warum betrachten
wir sie? Wer hat uns darüber erzählt?}




\begin{definition}
  Let $G$ be an undirected graph. A \emph{tree decomposition} of $G$
  is a tuple $(T,(X_t)_{t\in V(T)})$ where $T$ is an undirected tree
  and 
  \begin{itemize}
  \item $\bigcup_{t\in V(T)} X_t = V(G)$,
  \item for all $\{v,w\}\in E(G)$ there is some $t\in V(T)$ with
    $\{v,w\}\subseteq X_t$,
  \item for all $v\in V(G)$ the set $\{t\in V(T) \mid v\in X_t\}$
    induces a (connected) subtree of~$T$. 
  \end{itemize}
\end{definition}

\begin{definition}
Let $G$ be a directed graph. An \emph{oriented tree decomposition} is a pair
$(T,(X_t)_{t\in T})$ where $T$ is an orientation of an undirected
tree and for each $t\in V(T)$, $X_t \subseteq V(G)$ such that $E(G)$
can be partitioned into two, possibly empty, sets $E(G) = E^t \cup
E^{\shc}$ (the tree edges and the shortcut edges) and the following conditions hold.
\begin{enumerate}[(1)]
\item $(\bar T,(X_t)_{t\in T})$ is a tree decomposition of $(V(G),
  E^t)$.

\item\label{cover_shc} If $(u,v) \in E^{\shc}$, $u\in X_s$, and $v\in X_t$, then there is
  a path from $s$ to $t$ in~$T$.
\end{enumerate}

We say that an edge $e\in E(G)$ is \emph{covered by the tree} if $e$
is contained in some bag~$X_t$. Otherwise we say that $e$ is a
\emph{shortcut edge}.  The \emph{width} of an oriented tree
decomposition $(T,(X_t)_{t\in T})$ is $\max_{t\in T} |X_t|$. 

The \emph{\otw} $\oTW(G)$ of a graph $G$ is the minimum
width over all oriented tree decompositions of $G$.
\end{definition}

The following lemma states that all undirected edges are covered by
the tree. The proof is easy and we omit it.

\begin{lemma}\label{lemma_cover_undir_edges}
  Let $G$ be a graph. Let $(T,(X_t)_{t\in V(T)})$ be an oriented tree
  decomposition of~$G$ and let $E(G) = E^t \cup E^{\shc}$ be a
  corresponding partition of $E(G)$. If~$(v,u)\in E(G)$ and $(u,v) \in
  E(G)$, then $(u,v) \in E^t$ or $(v,u) \in E^t$.
  \end{lemma}

In the next lemma we give a normal form for oriented tree decompositions.

\begin{lemma}\label{lemma_normal_form_otdec}
  Let $G$ be a graph, let $(T,(X_t)_{t\in V(T)})$ be an oriented tree
  decomposition of $G$ and let $E^t\cup E^{\shc}$ be a corresponding
  partition of $E(G)$. Then there is an oriented tree decomposition of
  $G$ of the same width such that for all $(s,t)\in V(T)$, $X_s
  \not\subseteq X_t$.
\end{lemma}
\begin{proof}
  We construct a new decomposition by successively eliminating bags
  whose neighbours are their supersets. This suffices as by
  monotonicity if $X_s \subseteq X_t$, then $X_s\subseteq X_q$ for all
  $q\in V(T)$ on the path between~$s$ and~$t$ in~$\bar T$. Let
  $(s,t)\in E(T)$ be nodes with $X_s\not\subseteq X_t$. The new
  decomposition is $(T', (X_r)_{r\in T'})$ where $V(T') =
  (V(T)\setminus\{s\})$ and $E(T') = \bigl(E(T) \setminus \{(r,s),(s,r) \mid
  r\in T\}\bigr) \cup \{(t,t)' \mid (s,t')\in E(T)\} \cup \{(t',t)
  \mid (t',s) \in E(T)\}$. It is straightforward to check that all
  conditions of an oriented tree decomposition hold. Furthermore the
  width of the decomposition did not change.
\end{proof}

Our next goal is to compare \otw with \dw.

\begin{theorem}\label{thm:dw_otw}
 For all graphs $G$, $\DW(G) \le \oTW(G)$.
\end{theorem}
\begin{proof}
  Let $(T,(X_t)_{t\in V(T)})$ be an oriented tree decomposition of
  $G$. We argue that it is also a D-decomposition (of the same
  width). First, $\bigcup_{t\in V(T)}X_t = V(G)$, as $(\bar
  T,(X_t)_{t\in V(\bar T)})$ is a tree decomposition of~$\bar
  G$. Let~$S$ be a strongly connected set of vertices of~$G$ and
  assume that there are $s$, $q$ and $t$ in $V(T)$ with $X_s\cap S
  \neq \0$, $X_q \cap S = \0$, and $X_t \cap S\neq \0$ such that~$q$
  is on the path from~$s$ to~$t$ in~$\bar T$. We choose $s$, $q$,
  and~$t$ such that the path from $s$ to $t$ has minimum length.
  As~$S$ is strongly connected, there is a path~$P$ from some vertex
  in $X_s \cap S$ to $X_t \cap S$ in $G$. Let $(v,w)$ be the first
  edge of $P$ with~$v\in (P\cap X_s)\setminus X_t$ and $w\in (P\cap
  X_t)\setminus X_s$. (If $(v,w)$ does not exist, then there is a
  vertex~$u\in X_s\cap P \cap X_t$ and thus $u\in X_q$, as $q$ is
  between $s$ and $t$ in $\bar T$, but $P\subseteq S$, so $u\in X_q\cap
  S$, a contradiction.) By the choice of $s$, $q$ and~$t$, the edge
  $(v,w)$ is not covered by the tree (otherwise there would be an edge
  of the tree decomposition connecting $s$ and $t$, but there is
  another path from $s$ via $q$ to $t$). It follows that $(v,w)$ is a
  shortcut edge and hence all edges of the tree decomposition on the
  path from $s$ to $t$ are oriented from $s$ to $t$. By a symmetric
  argument we can show that all edges are oriented from $t$ to $s$, a
  contradiction.

%
\end{proof}

By \Cref{thm:dtw_cndtw} it follows that \dtw is bounded in \otw as
well.

\begin{corollary}
For every class~$\calG$ of graphs, if \otw is bounded on~$\calG$, then
\dtw is bounded on~$\calG$.
\end{corollary}

Recall that undirected binary trees with the additional upward
transitive closure from~\cite{BerwangerDawHunKreObd12} separate \dtw
and \dw from \dagw and \kw. It is easy to see that the \otw of such a
graph~$G$ is also small. The oriented decomposition tree has the same
shape as~$G$ and all tree edges are oriented upwards, so $\oTW(G) =
2$. Thus on some graphs \otw is bounded, but \dagw and \kw are not.
The next theorem shows that the opposite is also true: on some graph
classes with bounded \dtw, \dagw and \kw, \otw is not bounded by any
function, so the measures are incomparable in this sense.

\begin{theorem}\label{thm:dtw_otw}
  There is a family of graphs $G_n$ with $\dTWcops{G_n} = \DAGW(G_n)
  ={}$ $\KW(G_n) = \DW(G_n) = 2$ such that for each $k>2$ there is some $n$ with
  $\oTW(G_n) > k$.
\end{theorem}
\begin{proof}
  The graph $G_n$ is constructed inductively. Let $G_1^m$ be a single
  vertex. Then $G_{i+1}^m$ has a new root $r(G_{i+1}^m)$ with~$m$
  successors $v_1,\dots,v_m$. Each such successor $v_i$ is the root
  $r(G_i^m)$ of a copy of $G_i^m$ and has outgoing edges to all leaves
  of the copies $G_i^m$ rooted at $v_j$ with $j<i$. The construction
  of the graph $G^3_{i+1}$ from $G^3_{i}$ is shown in
  \Cref{fig_counterex_dtw_atw}. Finally, $G_n = G_n^n$.
\begin{figure}
\begin{center}
\begin{tikzpicture}

\node[vertex] (epsilon) at (0,0){} 
    child[sibling distance=2cm] { node[vertex] (0) {}
    }
     child[sibling distance=2cm] { node[vertex] (1) {}
    }
     child[sibling distance=2cm] { node[vertex] (2) {}
    };

  \draw (0) -- ++(0.5,-1) -- ++(-1,0) -- (0);
  \draw (1) -- ++(0.5,-1) -- ++(-1,0) -- (1);
  \draw (2) -- ++(0.5,-1) -- ++(-1,0) -- (2);
  \draw[-slim] (1) -- ++(-2.4,-1);
  \draw[-slim] (1) -- ++(-1.8,-1);
  \draw[-slim] (2) -- ++(-1.8,-1);
  \draw[-slim] (2) -- ++(-3.8,-1);
  \draw[-slim] (2) -- ++(-4.4,-1);
  \draw[-slim] (2) -- ++(-2.4,-1);

  \node at (0,-2.2){$G_i^3$};
  \node at (-2,-2.2){$G_i^3$};
  \node at (2,-2.2){$G_i^3$};

\end{tikzpicture}
\end{center}
\caption{The graph $G_{i+1}^3$.}
\label{fig_counterex_dtw_atw}
\end{figure}
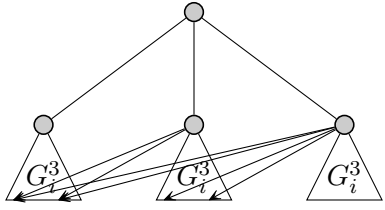
Formally, $V(G_n) = [n]^{\le n}$ is the set of words of length at
most~$n$ over the alphabet~$[n]$, $E(G_n) = E_t \cup E_r$ where $E_t =
\{(w,wa), (wa,w) \mid w\in [n]^{<n}, a\in[n]\}$ are edges forming the
tree and $E_r = \{(wa,wbv) \mid a,b,\in [n], b < a,  |v| = n - |w| +1\}$.

In the \dagw game two cops play from the top to the bottom of~$R$
along the path chosen by the robber. If the robber changes to a smaller
subtree, the last placed cop follows him to that subtree on its
root. Of course, one cop is unable to win. So $\DAGW(G) =2$ and for
the same reason also $\dTWcops{G} = 2$. The elimination order for the
\Kw is the depth-first search with choosing the right-most successor first.
The D-decomposition is the usual tree decomposition of a tree (every
bag contains two neighboured vertices and one bag contains only the root): edges
from $E_r$ do not destroy the conditions of a D-decomposition.

Now assume towards a contradiction that there is an oriented tree
decomposition $(T,(X_t)_{t\in V(T)})$ of~$G_n$ of width~$k$ where $n=k+3$.  Let
$E(G_n) = E^t \cup E^{shc}$ be the corresponding partition of edges
of~$G_n$. We first analyse the tree decomposition of $T_n \coloneqq
\overline{(V(G_n), E^t)}$. By Lemma~\ref{lemma_cover_undir_edges},
$E_t \subseteq E^t$.

For a connected subgraph $G'$ of an undirected graph~$G$ with tree decomposition 
$(T,(X_t)_{t\in V(T)})$, the restriction of $(T,(X_t)_{t\in V(T)})$ to $G$ is 
$(T',(X'_t)_{t\in V(T')})$ where~$T'$ is the subgraph of~$T$ induced by bags~$t$
with $X_t \cap V(G') \neq \0$ and $X'_t = X_t \cap V(G)$. Note that the 
restriction is a tree decomposition of~$G'$ of width at most the width of 
$(T,(X_t)_{t\in V(T)})$.

\begin{claim}\label{claim_natural_dec}
There is a subtree~$T'_n$ of~$T_n$ of (the same) depth~$n$ and such
that
\begin{itemize}
\item every node of $T'_n$ which is not a leaf has two children, and
\item the restriction $(T',(X'_n))$
of $(T,(X_t))$ to $T'_n$ is a natural decomposition, \ie up to isomorphism:
\begin{itemize}
  \item $T' = T'_n$,
  \item $X'_r = \{\epsilon\}$ (the root bag contains only the root of $T'_n$), 
and
  \item $X'_t = \{s,t\}$ where $s$ is the predecessor of~$t$ if~$t$ is not the 
root bag.
\end{itemize}
\end{itemize}
 \end{claim}
\begin{proof}[Proof (of Claim~\ref{claim_natural_dec}).]
  We use the characterisation of \tw by the \tw game. It is played on an
  undirected graph as the \dagw game. It is well known that a tree
  decomposition of width~$k$ induces a winning strategy for~$k+1$
  cops. In the first move they occupy the root bag. The robber chooses
  a subtree of the decomposition tree and the cops occupy the root of
  that subtree in the next move. Continuing in that way they finally
  capture the robber in a leaf bag.

  Let $\sigma$ be the strategy for~$k+1$ cops on $T_n$. For a
  vertex $v\in T_n$ let $T_v$ be the subtree of $T_n$
  rooted at~$v$. Consider a position of a play consistent with~$\sigma$ where 
the robber is
  in $T_v$ and $T_v$ is cop free. Then the following lemma holds.

\begin{claim}\label{claim_free_subtree}
  There are two children~$w_1$ and $w_2$ of~$v$ such that in any play
  from the current position that is consistent with~$\sigma$, for
  $i=1,2$, $w_i$ is the first vertex of $T_{w_i}$ occupied by a cop.
\end{claim}

Indeed, $v$ has more children than there are cops. If a cop is placed
in a subtree rooted at a child of~$v$, then there will be at least one
cop in that subtree until~$v$ is occupied (otherwise the
robber-\hspace{0cm}monotonicity is violated).

Now we define $T'_n$ in a top-\hspace{0cm}down manner. The root of
$T'_n$ is the root of $T_n$. Assume that a subtree of $T_n$ up to some
level is constructed. Let~$v$ be a current leaf and let~$w_1,w_2\in
T_n$ be the children of~$v$ whose existence is guaranteed by
Claim~\ref{claim_free_subtree}. Then~$v$ has two children
in~$T'_n$:~$w_1$ and~$w_2$. Then $(T'_n,(X'_t)_{t\in V(T'_n)})$ is a
natural decomposition. This proves the claim.
\end{proof}

Without loss of generality let $T'_n = (\{0,1\}^n, \{\{v,va\} \mid v\in 
\{0,1\}^n-1, a\in\{0,1\}\})$. Consider the edges $e_1 = (01,0^n)$ and 
$e_2 = (1,01^{n-1})$ in $E_r$. By the construction of $(T',(X'_t)_{t\in 
V(T')})$, the edge $e_1$ is not covered by the tree. Thus the 
orientation of $T'$ allows the path from~$1$ to $01^{n-1}$. In particular all 
edges on the path from $\epsilon$ to $01^{n-1}$ are oriented towards 
$01^{n-1}$, \ie the edge $\{0,01\}$ is oriented as $(0,01)$. The edge $e_2$ is 
not covered by the tree either, so the orientation allows the path from $01$ to 
$0^n$. In particular the edge $\{0,01\}$ is oriented as $(01,0)$, a 
contradiction.
\end{proof}

\begin{theorem}
For all graphs $G$ we have $\DS(G) \le \DW(G) \le \oTW(G)$.
\end{theorem}
\begin{proof}
  Let $(T,(X_t)_{t\in V(T)})$ be an oriented tree decomposition
  of~$G$ of width~$k$ and let $E(G) = E^t \cup E^{\shc}$ be a
  corresponding partition of $E(G)$. Let $\{s,t\}\in G(\bar T)$ be an
  edge of $\bar T$. Let $T^t$ be the maximal subtree of $\bar T$
  containing $t$, but not $s$ and $T^s$ the maximal subtree containing
  $s$, but not $t$. Let $X^t = (\bigcup_{q\in V(T^t)} X_q)\setminus
  X_s$ and $X^s = (\bigcup_{q\in V(T^s)} X_q)\setminus X_t$.  Let
  $(v,w)\in E(G)$. If $v\in X^t$ and $w\in X^s$, then $(v,w)$ is a
  shortcut edge and $(t,s) \in E(T)$ (and thus $(s,t)\notin E(T)$).

  Let $r$ be an arbitrary node of~$T$ that we declare to be the
  root. Let~$T_r$ be the orientation of~$T$ such that all edges are
  oriented away from~$r$.  We claim that $(T_r,(X_t)_{t\in V(T_r)})$ fulfils all
  requirements of a DS-decomposition except,
  possibly~(4). Note that $(T,(X_t)_{t\in V(T)})$ and
  $(T_r,(X_t)_{t\in V(T_r)})$ have the same width.
 Requirements~(1) and~(2) hold because $(\bar
  T,(X_t))$ is a tree decomposition. For~(3) assume for some
  $(s,t)\in E(T_r)$ that there is a path~$P$ that starts in $X^t$,
  leaves it and then returns to $X^t$ such that $P\cap X_s \cap X_t =
  \0$. Then there is an edge $(v,w)\in E(G)$ that goes from $X^t$ to
  $X^s$, so $(t,s) \in E(T)$ (and $(s,t)\notin E(T)$), and there is an
  edge that goes from $X^s$ to $X^t$, so $(s,t) \in E(T)$, a
  contradiction.
\end{proof}


\section{Conclusion}\label{sec_conc}

\subsection{The Relations between Widths and Cop Numbers}
The relations between \dtw, $\rmdTWcops{G}$ (the robber monotone cop
number in the component game), $\cmdTWcops{G}$, \dagw, weakly monotone
\dagw, \kw, \dw, and \otw are presented in \Cref{fig:scheme}. All
relations are considered in terms of boundedness. If~$a$ and~$b$ are
two measures from the above list, $a\le b$ means that there is a
function~$f\colon \bbN\to \bbN$ such that for all graphs~$G$, $a(G)
\le f(b(G))$. We write $a<b$ if $a\le b$ and there is a class~$\calG$
of graphs and a number $t\in\bbN$ such that for all $G\in \calG$,
$a(G) \le t$ and for all $s\in\bbN$ there is a graph $G\in \calG$ with
$b(G)>s$, \ie $a$ is bounded on~$G$ and~$b$ is not. We write $a=b$ if
there is a function $f\colon\bbN\to\bbN$ such that for all graphs~$G$,
$a(G) \le f(b(G))$ and $b(G) \le f(a(G))$. Finally $a\lessgtr b$ means
that there is a class of graphs on which $a$ is bounded and~$b$ is not and vice
versa: there is a class of graphs on which $b$ is bounded and $a$ is not.

\begin{figure}
\begin{center}
\begin{tikzpicture}

\node (dtw) at (0,0){dtw};
\node (rmdtw) at (0,-1){rmdtw};
\node (cmdtw) at (3,0){cmdtw};
\node (DAGw) at (6,0){DAG-w};
\node (wmDAGw) at (6,-1){wmDAG-w};
\node (Kw) at (9,0){K-w};
\node (Dw) at (6,-3){D-w};
\node (DSw) at (6,-4){DS-w};
\node (otw) at (9,-3){otw};

\node[label=above:\tiny \cite{JohnsonRobSeyTho01},rotate=90] (eq_dtw_rmdtw) at
(0,-0.5){$=$};

\node[label=above:\tiny Th. \ref{thm_non_cop_mon}] at (1.5,0) {$<$};

\node[label=above:\tiny \cite{BerwangerDawHunKreObd12}] at (4.5,0)
{$<$};

\node[label=below:\tiny Th. \ref{thm_main},rotate=90] at (6,-0.5) {$=$};

\node[label=above:\tiny Th. \ref{thm_kw_to_dagw}] at (7.5,0) {$\le$};

\node[label=below:\tiny {Th. \ref{thm_dw_game},
Th. \ref{thm_dw_game_fails}},rotate=-45] at (4.5,-2) {$\le$};

\node[label=above:\tiny {L. \ref{lemma_cmp_dw_dprimw}, \cite{Gruber08}},rotate=90]
at (6,-3.5) {$=$};

\node[label=left:\tiny {Th. \ref{thm_dw_game_fails},
  \cite{BerwangerDawHunKreObd12}},rotate=90] at (6,-2) {$\lessgtr$};

\node[label=below:\tiny {Th's \ref{thm:dw_otw}, \ref{thm:dtw_otw}}] at (7.5,-3) {$<$};

\node[label={[text width=0.985cm, text justified]below:\tiny {Th. \ref{thm:dtw_otw},
  \cite{BerwangerDawHunKreObd12}}},rotate=90] at (9,-1.5) {$\lessgtr$};

\node[label={[text width=0.985cm, text justified,label
  distance=-0.6em]-30:\tiny {Th. \ref{thm:dtw_otw},
    \cite{BerwangerDawHunKreObd12}}},rotate=-60] at (7.8,-1.5)
{$\lessgtr$};

\draw (Dw) .. controls (8,-5.5) and (12.8,-1.7) .. (Kw);

\node[label=right:\tiny {Th. \ref{thm:dtw_otw},
  \cite{BerwangerDawHunKreObd12}}] at (7.5,-4.2) {$\lessgtr$};

\end{tikzpicture}
\end{center}
\caption{The relations between different measures. }
\label{fig:scheme}
\end{figure}
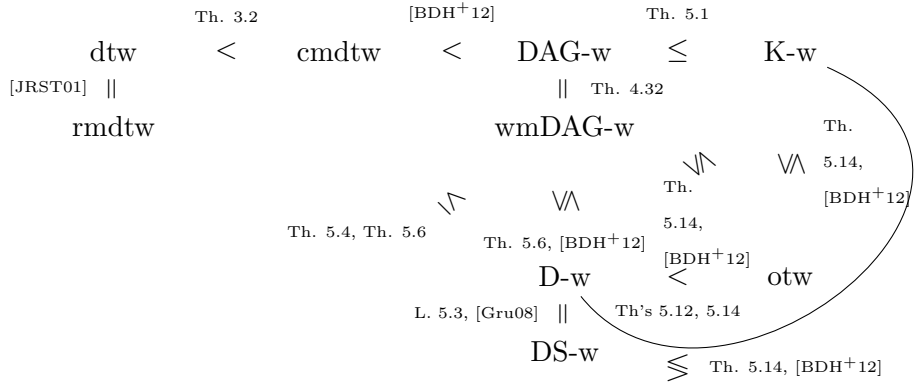

\subsection{Future Work}

We believe that the quadratic blowup in the number cops when we change
from a weakly monotone winning strategy for the cops to a strongly
monotone strategy can be reduced to a linear one. The largest
ratio between the numbers of needed cops in the weakly monotone and
the strongly monotone cases we are aware of is $4/3$ from the examples
by Kreutzer and Ordyniak. However, the induction on the number of cops
needed by the weakly monotone strategy seems to enforce the use of
quadratically many cops. It would be interesting to achieve a linear
upper bound or to find better lower bound than $4/3$.

Another topic for the future work is the gap between the
non-{}monotone and weakly monotone case. This question seems to be the
most interesting in this area. It would be also important to determine
whether the inequality $\DAGW \le \KW$ should be strict or an
equality.

It is not known whether \dagw can be decided in \nptime as \kw. Neither is
known whether \kw can be computed in time $n^{O(k)}$ where~$n$ is the
size of the given graph and~$k$ is its \kw.



\bibliographystyle{alpha}
\newcommand{\etalchar}[1]{$^{#1}$}

\end{document}